\documentclass[floatfix,twocolumn,superscriptaddress, prl,aps]{revtex4-2}
\usepackage[utf8]{inputenc}
\usepackage{amsmath}
\usepackage{amsfonts}
\usepackage{amssymb,amsthm}
\usepackage[subrefformat=parens,labelformat=parens,caption=false]{subfig}
\usepackage{graphicx}

\usepackage{multirow}
\usepackage[]{hyperref}
\usepackage{bbold}
\usepackage{cancel}
\hypersetup{
  colorlinks = true,
  urlcolor = magenta,
  linkcolor = red,
  citecolor = blue
}
\usepackage[normalem]{ulem}
\usepackage{braket}
\usepackage[all]{hypcap}
\usepackage{url}

\newtheorem{theorem}{Theorem}
\newtheorem{lemma}[theorem]{Lemma}

\usepackage[]{hyperref}
\usepackage{bbold}
\usepackage{bm}
\usepackage{cancel}
\newcommand{\eq}[1]{\begin{align}#1\end{align}}
\hypersetup{
  colorlinks = true,
  urlcolor = magenta,
  linkcolor = red,
  citecolor = blue
}

\usepackage{graphicx}
\graphicspath {{figures/}}
\begin{document}

\title{Universal scattering with general dispersion relations
}
\author{Yidan Wang}
 \affiliation{Joint Quantum Institute, NIST/University of Maryland, College Park, Maryland 20742 USA}
\author{ Michael J. Gullans}
 \affiliation{Joint Center for Quantum Information and Computer Science, NIST/University of Maryland, College Park, Maryland 20742 USA}
 \author{Xuesen Na}
 \affiliation{Department of Mathematics, University of Maryland, College Park, Maryland 20742, USA}
 \author{{\color{black}Seth Whitsitt}}
 \affiliation{Joint Quantum Institute, NIST/University of Maryland, College Park, Maryland 20742 USA}
 \affiliation{Joint Center for Quantum Information and Computer Science, NIST/University of Maryland, College Park, Maryland 20742 USA}

\author{Alexey V. Gorshkov}
\affiliation{Joint Quantum Institute, NIST/University of Maryland, College Park, Maryland 20742 USA}
 \affiliation{Joint Center for Quantum Information and Computer Science, NIST/University of Maryland, College Park, Maryland 20742 USA}

\begin{abstract}
Many synthetic quantum systems allow particles to have  dispersion relations that are neither linear nor quadratic functions.   
{\color{black} Here,  we explore single-particle scattering in general spatial dimension $D\geq 1$ when the density of states diverges at a specific energy. To illustrate the underlying principles in an experimentally relevant setting, we focus on waveguide quantum electrodynamics (QED) problems (i.e.~$D=1$) with dispersion relation $\epsilon(k)=\pm |d|k^m$, where $m\geq 2$ is an integer.} 
   For a large class of these problems for any positive integer $m$, we rigorously prove that when 
   there are no bright zero-energy eigenstates,  the $S$-matrix evaluated at an energy  $E\to 0$ converges to a universal limit that is only dependent on $m$.
   We also give a generalization of a key index theorem in quantum scattering theory known as Levinson's theorem---which relates the scattering phases to the number of bound states---{\color{black}to waveguide QED scattering} for these more general dispersion relations. {\color{black} We then extend these results to  general integer dimensions $D \geq 1$, dispersion relations $\epsilon(\bm{k}) =  |\bm{k}|^a$ for a $D$-dimensional momentum vector $\bm{k}$ with any real positive $a$,  and separable potential scattering.}
\end{abstract}

\maketitle

The quantum mechanical scattering of few-body  systems remains a challenging theoretical problem.  Even at low incoming energies, nonperturbative effects render a general solution out of reach. 
A common workaround is based on effective field theory whereby low-energy scattering is described in terms of a few  parameters such as the scattering length $a_0$ and the effective range $r_0$ \cite{bethe1949theory, braaten2006universality, braaten2001three}.  
When $a_0\gg r_0$, the system is in the unitarity limit where the universal physics of Efimov states \cite{efimov1973energy, braaten2006universality, braaten2001three} and unitary Fermi gases \cite{thomas2005virial, nishida2008universal, nascimbene2010exploring} can emerge.  
Another approach where general results can be obtained is by studying the analytic structure of the $S$-matrix at low energies.  One striking result in this context is the simple effect of dimensionality on scattering theory. 
Two particles with short-range interactions perfectly reflect off each other at the threshold in one dimension (1D),  while they transmit without seeing each other in higher dimensions.  This effect arises because the density of states diverges at the threshold as {\color{black}$1/\sqrt{E}$} in 1D, but stays finite in higher dimensions.

Recent experimental progress in synthetic quantum matter allows for broad control of dispersion relations. 
 One class of such systems consists of tunable periodic structures, including photonic crystal waveguides \cite{hughes2004enhanced, Joannopoulos:08:Book, Hung2013,Goban2014, Goban2015, Hood2016, Lodahl2015}, twisted bilayer graphene \cite{tarnopolsky2019origin, cao2018unconventional}, superconducting qubit arrays \cite{VanLoo2013, Devoret2013, sundaresan2019interacting}, atomic arrays \cite{endres2016atom, barredo2016atom, madjarov2019atomic}, and trapped-ion spin chains \cite{porras2004effective, debnath2018observation}.  Another class is polaritonic \cite{kittel1976introduction, band2006light} or spin-orbit coupled \cite{campbell2011realistic, lin2011spin} systems, where the dispersion relation can be tuned \textit{in situ} by external fields  \cite{Fleischhauer00, mahan2013many, Peyronel2012Jul, Firstenberg2013}.
 In principle, the density of states at the scattering threshold can be tuned to diverge faster than it does for quadratic dispersion relations.
This opens up the door to studying the implications of a more general density of states without changing the dimension of the system. 
Recently, there is a growing interest in the study of general dispersion relations in condensed matter systems, where divergent electronic density of states is referred to as a high-order Van Hove singularity \cite{isobe2019supermetal, yuan2019magic,yuan2020classification}. In particular, power-law-divergent density of states near the Fermi level leads to nontrivial metallic states termed supermetals \cite{yuan2019magic}.

 In this Letter, we explore the physics of divergent density of states from the perspective of scattering theory. We illustrate that, when a particle has a divergent density of states at a certain energy, its scattering matrix has a nontrivial universal limit that depends on the rate of the divergence. {\color{black} In the main text of this Letter}, we study single-particle scattering of photon-{\color{black}emitter} models in 1D ($D=1$) with a dispersion relation  $\epsilon(k)=\pm |d| k^m$, where $m$ is a positive integer. 
Notably, {\color{black} when $m$ is even, these emitter scattering models describe scattering for incoming frequencies near the band edge of photonic crystal waveguides coupled to atoms \cite{Hung2013} or quantum dots \cite{hughes2004enhanced}.
}
We discover that the $S$-matrix can take different universal limits $\lim_{E\rightarrow 0} S(E)$  for different values of $m$.
 The total reflection at the threshold for a quadratic dispersion relation is an example of such universal behavior corresponding to $m=2$. In general, there may be multiple classes of universal behaviors in the $S$-matrix corresponding to each $m$, depending on the properties of interactions at $k=0$. 
In  this Letter, we consider a physically natural class of interactions and characterize the universal behavior for each $m$.   
We also extend a key index theorem in scattering theory known as Levinson's theorem---which relates the  scattering phases to the number of bound states
\cite{Levinson1949, jauch1957relation, ida1959relation, wright1965, atkinson1966levinson, ma1985levinson,  barton1985levinson, Poliatzky1993, dong1998relativistic,  dong2000levinson,  ma2006levinson}---to the class of models considered in this Letter with these more general dispersion relations.
 {\color{black} To demonstrate the generality of our methodology, in the Supplemental Material \cite{supp}, we extend our discussions to separable potential scattering, general integer dimensions $D\geq 1$ and dispersion relations $\epsilon(k)=|\bm{k}|^a$, where  $\bm{k}$ is a $D-$dimensional momentum vector and $a$ is any positive real number. The extension of our single-emitter results to spin-boson models is given in an upcoming work \cite{Whitsitt2021}. These spin-boson models generalize the waveguide quantum electrodynamics (QED) models introduced below by including emitter-photon interaction terms beyond the rotating wave approximation---thereby, illustrating the relevance of our results in the many-body regime of waveguide QED. }

{\color{black}\emph{Waveguide QED.}}---In many synthetic quantum systems, particles propagating in a 1D channel are scattered by {\color{black}emitters} such as atoms, quantum dots, or superconducting qubits. 
The {\color{black}emitters} are often coupled to the environment, which adds dissipation to the system composed of the {\color{black}emitters} and the 1D channel.  {\color{black}Such models are broadly referred to as waveguide QED models.}
Since we are interested in the scattering processes with a single photon coming in and a single photon going out,  it suffices to use a non-Hermitian effective quadratic Hamiltonian  
\eq{
H&=H_0+V,\\
H_0&=\int_{-\infty}^{+\infty}dk \ \epsilon(k)C^\dagger(k)C(k)+\sum_{i,j=1}^N K^R_{ij}b_{i}^{\dagger}b_{j}, \\
V&=\int_{-\infty}^{+\infty} dk \left[\sum_{i=1}^N V_i(k)C(k)  b^{\dagger}_{i}+\text{h.c}\right],\label{EqEffHamil}  
}
where the bare Hamiltonian $H_0$ consists of the freely propagating particles, while the interacting {\color{black}emitters} are indexed by $i=1,2,\dots N$. $V$ describes the quadratic interaction between the particles and the {\color{black}emitters}. {\color{black} Through controlling  the lattice structures of the photonic crystal waveguide, the rate at which the density of states diverges at a particular energy can be fine-tuned.} Since we are discussing single-particle scattering with bounded-strength interactions, only local spectral properties of the dispersion relation matter, and our results are insensitive to the detailed behavior of the dispersion far away from the threshold energy. 
In this Letter, we focus on the dispersion relation $\epsilon(k)=\sigma|d|k^m$, where $\sigma=\pm1$, $|d|$ is a positive constant, and $m$ is a positive integer.  {\color{black} The case of $m=1$ corresponds to a linear dispersion relation and has a non-universal scattering matrix in the limit of zero energy \footnote{{\color{black}Note, for dispersions relations in 1D of the form $\epsilon(k) = |k|^a$, the $S$-matrix obtains a universal value for any positive real $a$ \cite{supp}.  These non-analytic dispersion relations have a trivial universal limit for the $S$-matrix when $a \le 1$. For $a> 1$, they have similar universal behavior of the $S$-matrix as the positive even integer $m$  cases of  $\epsilon(k) = \sigma |d| k^m$ studied in the main text.}}. For this reason, we assume $m\geq 2$ in the discussion below. } When $\sigma=\pm1$ and $m$ is even, $\epsilon(k)$ can be understood as the lowest-order approximation of a dispersion relation around its local minima/maxima, after a change of reference points for both energy and momentum.  
Depending on whether we are considering bosons scattered by bosonic {\color{black}emitters} or fermions scattered by fermionic {\color{black}emitters},  we have either commutation or anti-commutation relations: $[C(k),C^\dagger(k')]_{\pm}=\delta(k-k'), [b_i, b^\dagger_j]_{\pm}=\delta_{ij}$. 
 $K^R_{ij}$ represents the matrix element of the $N\times N$ matrix $\boldsymbol{K}^R$; $\boldsymbol{K}^R$ is 
 the only non-Hermitian term in the Hamiltonian: the Hermitian $\boldsymbol{A}$ and anti-Hermitian $i \boldsymbol{B}$ components of $\boldsymbol{K}^R=\boldsymbol{A}+i\boldsymbol{B}$ represent, respectively, the coherent and incoherent interactions among the emitters.  
 $\boldsymbol{K}^R$ is dissipative when $\boldsymbol{B}$ is non-positive and nonzero.    


For convenience, we introduce a vector function
$|v_k\rangle=[V_{1}(k),\dots, V_{N}(k) ]^T$, with corresponding  basis states given by the emitter excitations $\{ b^\dagger_1\ket{0,g}, \dots, b^\dagger_N\ket{0,g}\}$, where $\ket{0,g}$ is the ground state with zero excitation.  
In the most generic scenario, $V_i(k)$ for different {\color{black}emitters} are independent of each other.
Here, we consider the case where $\ket{v_k}$ can be written as $\ket{v_k}=V(k)\ket{u}$.  
We further assume $V(k)$ is continuous at $0$ and $V(0)\neq 0$. 
  Under this constraint, the only relevant vector around $k=0$ is $\ket{u}$, and effectively, there is only a single relevant ``degree of freedom" in the {\color{black}emitter} vector space at $k=0$. We then show that the zero-energy scattering behavior for multiple {\color{black}emitters}  can be reduced to the behavior for $N=1$. 
 As a result, we are able to obtain a complete classification of the universal low-energy scattering behavior in these models.

\emph{Universal scattering.}---{\color{black}We start with a discussion that applies to the case of general $\ket{v_k}$.} The $S$-matrix for a single particle is defined through the incoming and outgoing scattering eigenstates $|\psi_k^\pm\rangle$, where the superscript  $\pm$  specifies the boundary conditions of the scattering states.  
The $S$-matrix element from one single-particle scattering state $k$ to another $k'$ is  $\mathcal{S}(k,k')=\langle \psi^-_{k'}|\psi^+_{k}\rangle$. 
To explain the universal behavior of the $S$-matrix, it is useful to write down its  relation to the on-shell T-matrix: 
\eq{
\mathcal{S}(k,k')&=\delta(k\!-\!k')-2\pi i \delta[\epsilon(k)\!-\!\epsilon(k')] T(E\!+\!i0^+, k,k'),\label{eqST}
}
where  $0^+/0^-$ represents an infinitesimal positive/negative real number and $E=\epsilon(k)$. 
For dispersion relation $\epsilon(k)=\sigma|d|k^m$ with even $m$, there are two degenerate momenta $k_1(E),k_2(E)$ corresponding to any energy $E > 0$ ($E < 0$) for $\sigma = + 1$ ($\sigma =-1$). 
We can define a $2\times 2$ matrix $\boldsymbol{S}(E)$ by picking out the scattering amplitudes between degenerate momenta: 
\eq{
S_{\alpha\beta}(E)=\delta_{\alpha\beta}- 2\pi i\, \frac{ T[E\!+\!i0^+,k_\alpha(E), k_\beta(E)]}{\sqrt{|\epsilon'[k_\alpha(E)]\epsilon'[k_\beta(E)]|}},\label{eqSgvT}
}
where  $\alpha,\beta\in \{1,2 \}$ and {\color{black} the prefactor  $|\epsilon'[k_\alpha(E)]\epsilon'[k_\beta(E)]|^{-1/2}$ comes from $\delta[\epsilon(k)-\epsilon(k')]$ in Eq.\ \eqref{eqST}.} When $m$ is odd, we can define $\boldsymbol{S}(E)=S(E)$ as a single complex number, given by Eq.\ \eqref{eqSgvT} when  $k_\alpha(E)=k_\beta(E)=k(E)$ is the momentum corresponding to energy $E$. 
If the Hamiltonian is Hermitian, $S(E)$ is unitary.

In 1D scattering, the matrix $\boldsymbol{S}(E)$ directly describes the transmission and reflection between degenerate momenta and is often used instead of the function $\mathcal{S}(k,k')$. 
When $E\rightarrow 0$, $|\epsilon'[k_\alpha(E)]\epsilon'[k_\beta(E)]|^{-1/2}$ diverges. Since $|S_{\alpha\beta}(E)|\leq 1$, $T(E+i0, k_\alpha, k_\beta)$ in Eq.\ \eqref{eqSgvT} must approach zero to cancel  the divergence, which is the key behind the universal behavior of $\boldsymbol{S}(E)$.

To proceed further, we note that the Lippmann-Schwinger equations for this {\color{black}emitter} scattering model have a simple analytic structure.  
As a result, we can write down the single-particle  $T$-matrix $T(\omega, k,k') $ in terms of the Green's function of the {\color{black}emitters} $\boldsymbol{G}(\omega)$, which is a finite-dimensional matrix \cite{suhl1965dispersion}:
\eq{
T(\omega, k,k')&=\langle v_{k'}|\boldsymbol{G}(\omega)|v_{k}\rangle,\label{eqTG}\\
\boldsymbol{G}(\omega)&=\frac{1}{\omega\mathbb{1}_N-\boldsymbol{K}^R-\boldsymbol{K}(\omega)},\label{eqGK}\\
\boldsymbol{K}(\omega)&=\int_{-\infty}^{+\infty} dk \frac{|v_k\rangle \langle v_k|}{\omega-\epsilon(k)} \label{eqKdefi},
}
where $\mathbb{1}_N$ is an $N\times N$ identity matrix.
Equations  \eqref{eqTG}-\eqref{eqKdefi} hold for general photon-{\color{black}emitter} couplings where  $V_i(k)$ are independent functions for different {\color{black}emitters}. 
There are two mathematical conditions on $V_i(k)$ that are necessary for the integral in Eq.\ \eqref{eqKdefi} to be well-defined at any complex $\omega \neq 0$ outside the continuum spectrum \footnote{{\color{black} In {\color{black}emitter} scattering, it is natural to define the continuum spectrum to not include $0$}}.
First, we require that $V_i(k)$  is a locally square-integrable complex function on the real line. Second, to ensure that no ultraviolet divergences are present in the model, we impose a restriction on the large-$k$ behavior of $V_i(k)$:   when $k\rightarrow \pm \infty$, there exist $\gamma>1$ such that  $|V_i(k)|^2=o(|k|^{m-\gamma})$.
Each element of the $N\times N$ matrix $\boldsymbol{K}(\omega)$ is an analytic function on the complex plane with a branch cut along the continuum spectrum. $\boldsymbol{K}(\omega=E+i0^+)$ can be understood as describing effective interactions between {\color{black}emitters} induced by the 1D channel.

To understand the properties of $T(E+i0)$ close to $E=0$, we need to understand the behavior of $\boldsymbol{K}(\omega)$ around $\omega=0$. 
We can show that the value of $\boldsymbol{K}(\omega)$ around $\omega=0$ is decided by the dispersion relation and $V(0)\ket{u}$. 
Define $L(\omega)$ as the integral over the free-particle propagator:
\eq{
L(\omega)&=\int_{-\infty}^{+\infty}dk \frac{1}{\omega-\epsilon(k)}.
 \label{eqLsum0}
}
We see that, when $\omega\rightarrow 0$, $ L(\omega)^{-1}\frac{1}{\omega-\epsilon(k)}$ as a function of $k$ diverges at $k=0$ and vanishes everywhere else. In addition, $\int_{-\infty}^{+\infty} dk\ L(\omega)^{-1}\frac{1}{\omega-\epsilon(k)}=1$ by definition of $L(\omega)$. 
Hence,  it follows from a standard result in functional analysis attributed to Toeplitz \cite{PeterLax}  that 
$\lim_{\omega\rightarrow 0} L(\omega)^{-1}\frac{1}{\omega-\epsilon(k)}=\delta(k) $.
 Using the condition that $\ket{v_k}=V(k)\ket{u}$ is continuous at $k=0$ and the definition of $\boldsymbol{K}(\omega)$ in Eq.\ \eqref{eqKdefi}, we have
 \eq{
 \lim_{\omega\rightarrow 0}L^{-1}(\omega)\boldsymbol{K}(\omega)= |V(0)|^2\ket{u}\bra{u}. \label{eqLm1K}
 }
  When the {\color{black}emitter} region consists of a single site,
$\boldsymbol{K}^R = K^R$ is a complex number and  Eq.\ \eqref{eqLm1K} becomes $\lim_{\omega\rightarrow 0}L^{-1}(\omega)\boldsymbol{K}(\omega)= |V(0)|^2 $. Using Eqs.\ \eqref{eqTG} and \eqref{eqGK}, we then have  
\eq{
\lim_{\omega\rightarrow 0}L(\omega)T(\omega,k,k')
= -\frac{V^*(k')V(k)}{|V(0)|^2},\label{eqTL}
}
which is no longer dependent on $K^R$  because $\lim_{\omega\rightarrow 0}L^{-1}(\omega)K^R=0$.
Although Eq.\ \eqref{eqTL} is derived for the case of $N=1$, we show through a rigorous mathematical analysis in the Supplemental Material \cite{supp} that Eq.\ \eqref{eqTL} holds as long as the Hamiltonian does not support a ``bright'' zero-energy eigenstate,  defined as a zero-energy eigenstate that has a non-zero {\color{black}emitter} and photonic amplitude.
These bright states are distinguished from ``dark''  states that have only a nonzero photonic amplitude and rather generically arise at zero-energy in these models.  The proof of Eq. \eqref{eqTL} for $N>1$ is the main technical result of this Letter as it underlies both the universal scattering results and our proof of Levinson's theorem.

When we evaluate the $S$-matrix in the limit $E\rightarrow 0$ using Eq.\ \eqref{eqSgvT},   $k_\alpha(E),k_\beta(E)$ in the T-matrix are both  sent to $0$. Using Eq.\ \eqref{eqTL} and the condition that $V(k)$ is continuous at $k=0$, we have
\eq{
\lim_{E\rightarrow 0}L(E+i0^+)T[E+i0^+,k_\alpha(E),k_\beta(E)]
= -1,\label{eqTL1}
}
which shows that the on-shell $T$-matrix in the zero-energy limit is independent of the details of the interaction and fully determined by the dispersion relation; this is the reason behind the universal limit of the $S$-matrix when $E\rightarrow 0$.  In the Supplemental Material \cite{supp}, we  evaluate Eq.\ \eqref{eqLsum0} and obtain the $m$-dependent value of $L(\omega)$:
\eq{
L(\omega)
=-\pi i \kappa_m  \rho(|\omega|)\exp\left(-i\theta\frac{m-1}{m}\right),\label{eqLomega}
}
where the complex frequency $\omega$ is parameterized in  polar coordinates as $ \omega=\sigma\exp(i\theta)|\omega|$,  
 and $\rho(|\omega|)= \frac{2}{m|d|^{1/m}}|\omega|^{-1+1/m}$ corresponds to the density of states at energy $E=|\omega|$. 
For even $m$, $\kappa_m=\frac{2}{1-\mu^2}$ with $\mu=\exp(i\pi/m)$, while  $L(\omega)$ has a branch cut along the continuum spectrum $(0,+\infty)$ for $\sigma=+1$ or $(-\infty, 0)$ for $\sigma=- 1$. 
For odd $m$,  $\kappa_m=-\frac{1}{\mu-1}$ for $\theta\in (0,\pi)$ and $\kappa_m=-\frac{1}{\mu(\mu-1)}$ for $\theta\in (\pi,2\pi)$, while $L(\omega)$ has a branch cut along the real line. {\color{black}For both even and odd $m$, $L(\omega)$ diverges at the rate of density of states $\rho(|\omega|)$ when $\omega$ approaches $0$.

\begin{figure}
\subfloat[odd $m$]{
    \includegraphics[width=0.95\linewidth]{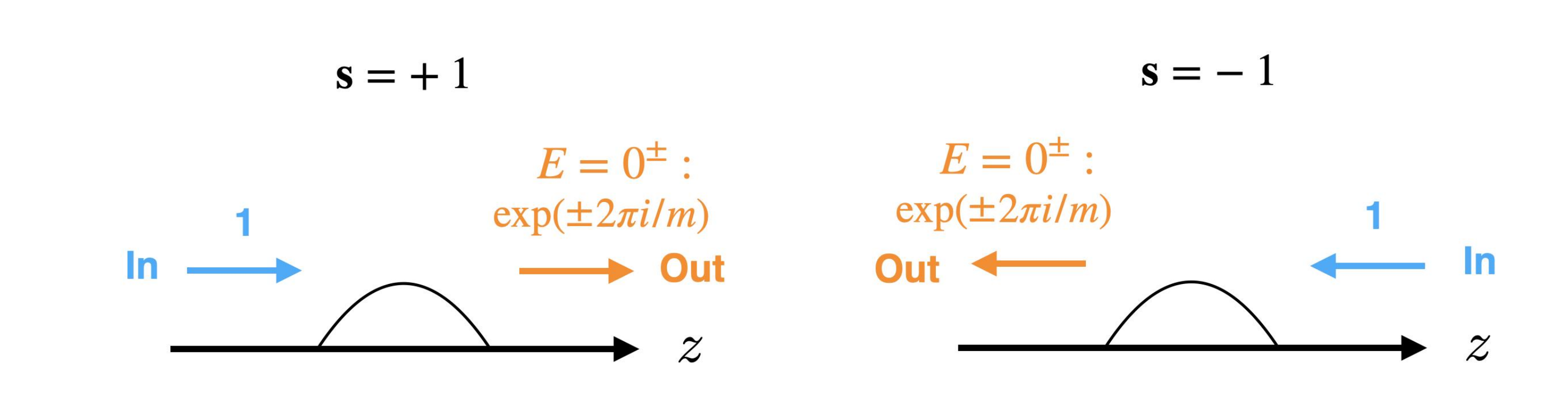}  \label{subfig:odd}
    }
\newline
 \subfloat[even $m$]{
    \includegraphics[width=0.95\linewidth]{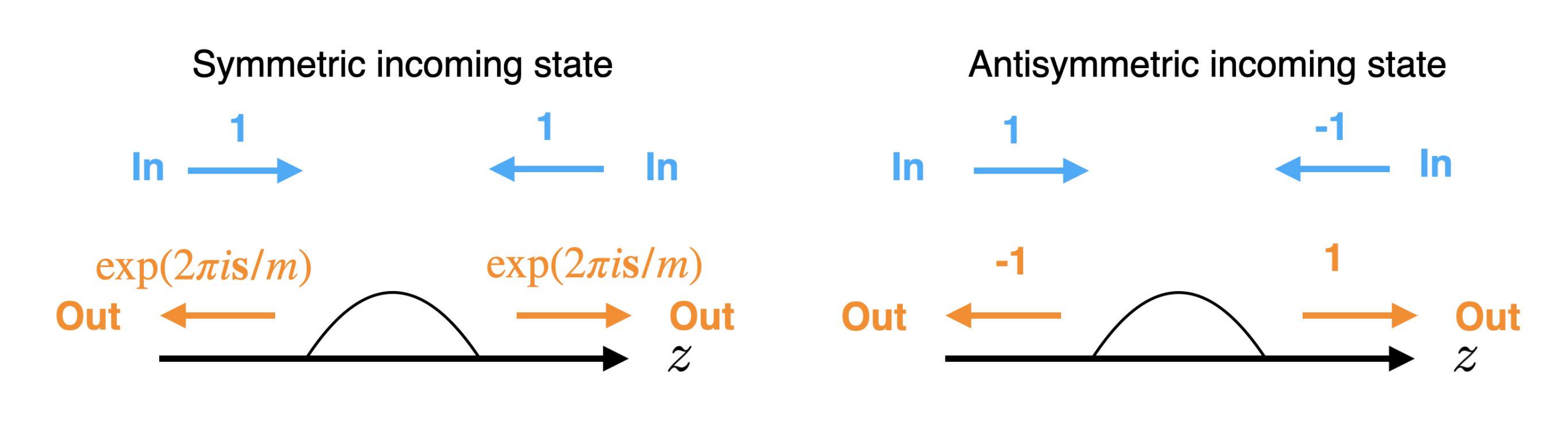}  \label{subfig:even}
    }
 \caption{
Illustration of 1D scattering ($z$ is a spatial coordinate) near zero energy for dispersion relation $\epsilon(E)=\sigma|d|k^m$ with $\sigma=\pm 1$. Panel (a) is for odd $m$, where the scattering matrix is a single transmission coefficient dependent on $m$ and the sign of energy is $E=0^\pm$.
Panel (b) is for even $m$, where the scattering matrix is a $2\times 2$ matrix. 
The eigenstates of the scattering matrix are the symmetric and antisymmetric incoming states, with eigenphases $\exp(2\pi i\sigma/m )$ and $1$, respectively.  }
\end{figure}

Now, we are ready to evaluate the  limit of the $S$-matrix at zero energy.   When $m$ is odd,  energy $E$ can approach $0$ from both above and below: $E\rightarrow 0^\pm$. When $m$ is even, $E$ can only approach $0$ from one side: $E\rightarrow 0^+$ when  $\sigma=+1$ or $E\rightarrow 0^-$ when  $\sigma=-1$. 
 Taking the limit $E\rightarrow 0^\pm$ in Eq.\ \eqref{eqSgvT} for the respective cases properly, we have
 \eq{
  \lim_{E\rightarrow 0^\pm }S_{\alpha\beta}(E) 
=\delta_{\alpha\beta}+\lim_{E\rightarrow 0^\pm } 2\pi i \rho(|E|) L^{-1}(E+i0^+), \label{EqlimSE0pm}
}
where we have used Eq.\ \eqref{eqTL1} and the observation that $\lim_{E\rightarrow 0}|\epsilon'(k_\alpha(E))\epsilon'(k_\beta(E))|^{1/2}\rho(E)=1$.}
Using Eqs.\ \eqref{eqLomega} and \eqref{EqlimSE0pm}, we find for odd $m$ 
\eq{
\lim_{E \rightarrow 0^\pm}S(E)=\exp(\pm \pi i/m),\label{eqSodd}
}  
as illustrated in Fig.\ \subref*{subfig:odd}. 
For even $m$, we find that the $S$-matrix
\eq{
\lim_{E\rightarrow \sigma0^+}\boldsymbol{S}(E) =\exp(\sigma i\pi/m)\begin{bmatrix}
\cos(\pi/m)& \sigma i\sin(\pi/m)\\
\sigma i\sin(\pi/m)&\cos(\pi/m)
  \end{bmatrix}, \label{eqSeven}
  }
is symmetric in the basis of degenerate momenta $\{|k_1=0^+\rangle, |k_2=0^-\rangle\}$. 
The symmetric eigenstate $|\psi_{\text{s}}\rangle=\frac{1}{\sqrt{2}}(1, 1)^T$ has an eigenphase $\exp(i\pi \sigma/m)$, while the antisymmetric eigenstate   $|\psi_{\text{a}}\rangle=\frac{1}{\sqrt{2}}(1, -1)^T$ has a trivial eigenphase $1$. 
The scattering of the symmetric and antisymmetric incoming states near zero energy is illustrated in Fig.\ \subref*{subfig:even}.
  For quadratic dispersion $\epsilon(k)=|d|k^2$, we recover the well-known total reflection:
  \eq{
   \lim_{E\rightarrow 0^+}\boldsymbol{S}(E)= \begin{bmatrix}
0&-1\label{EqS2}\\
-1 &0
  \end{bmatrix}.
}

{\color{black}The relation between the universal behavior of the S-matrix and the dispersion relation also applies to other types of interactions. In the Supplemental Material, we show that Eqs.\ \eqref{eqSodd} and  \eqref{eqSeven} also hold for separable potential scattering. In addition, we generalize our results to arbitrary integer dimension $D\geq 1$ and dispersion relations $\epsilon(k)=|\bm{k}|^a$, where $a > 0$ is not required to be an integer. In these cases, we demonstrate that the determinant of the S-matrix reaches a universal limit dependent only on $a/D$. }

\emph{Levinson's theorem.---}Levinson's theorem relates the quantized scattering phase to the number of bound states in the system. 
In the literature, the theorem has been discussed in various Hermitian systems and various dimensions \cite{Levinson1949, jauch1957relation, ida1959relation, 
wright1965, atkinson1966levinson, ma1985levinson,  barton1985levinson, Poliatzky1993, dong1998relativistic,  dong2000levinson,  ma2006levinson}, where the  dispersion relation close to the scattering threshold is always quadratic.  
In our recent work, we generalized Levinson's theorem to 1D {\color{black}emitter} scattering, where dissipation is present and the dispersion relation is linear at all $k$ \cite{wang2018single}.  
In that case, there is no well-defined scattering threshold. When we consider {\color{black} dispersion relations $\epsilon(k)=\sigma|d|k^m$ with the class of photon-{\color{black}emitter} couplings $\ket{v_k}=V(k)\ket{u}$, the $S$-matrix can take different universal limits at zero energy, dependent on the value of the integer $m\geq 2$ [see Eqs.\ \eqref{eqSodd} and \eqref{eqSeven}].}
This leads to a modification to Levinson's theorem, as we illustrate in the remainder of this Letter.

\begin{figure}
\subfloat[$\epsilon=\pm |d|k$]{
    \includegraphics[width=0.4\linewidth]{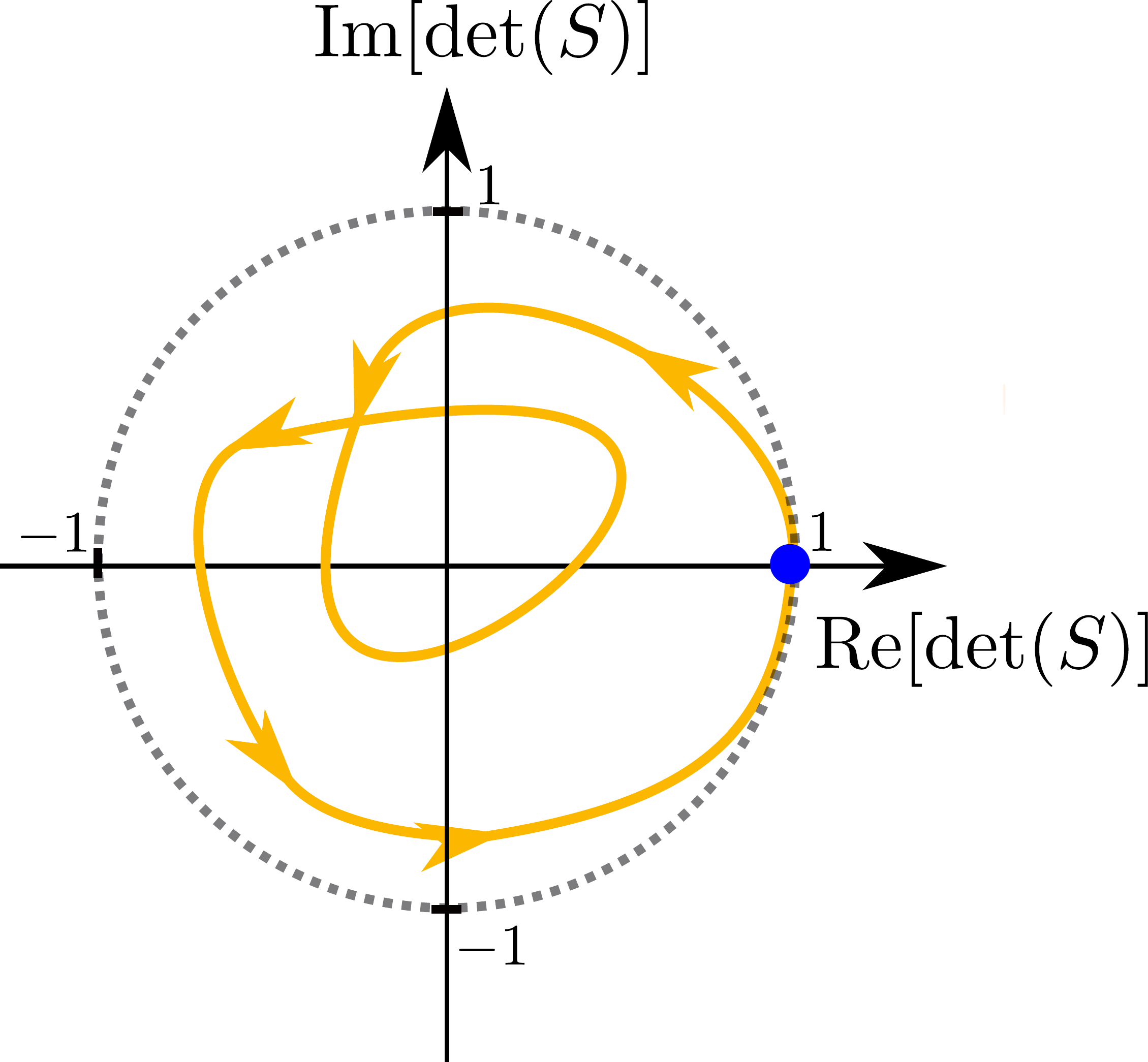}  \label{subfig:lineardetS}
}
     \subfloat[$\epsilon=|d|k^2$]{
    \includegraphics[width=0.4\linewidth]{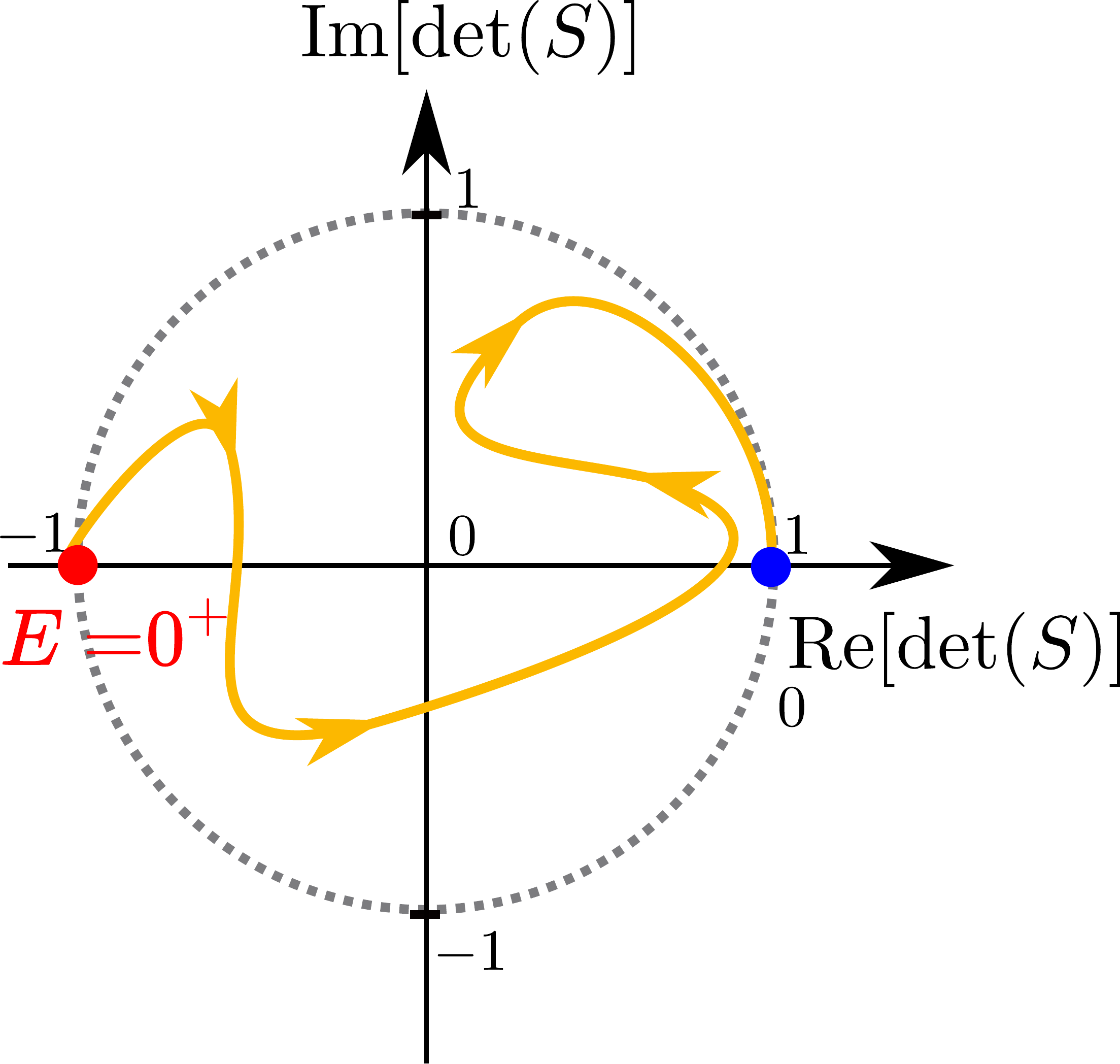}  \label{subfig:quadraticdetS}
}

    \subfloat[$\epsilon=|d|k^6$]{
    \includegraphics[width=0.4\linewidth]{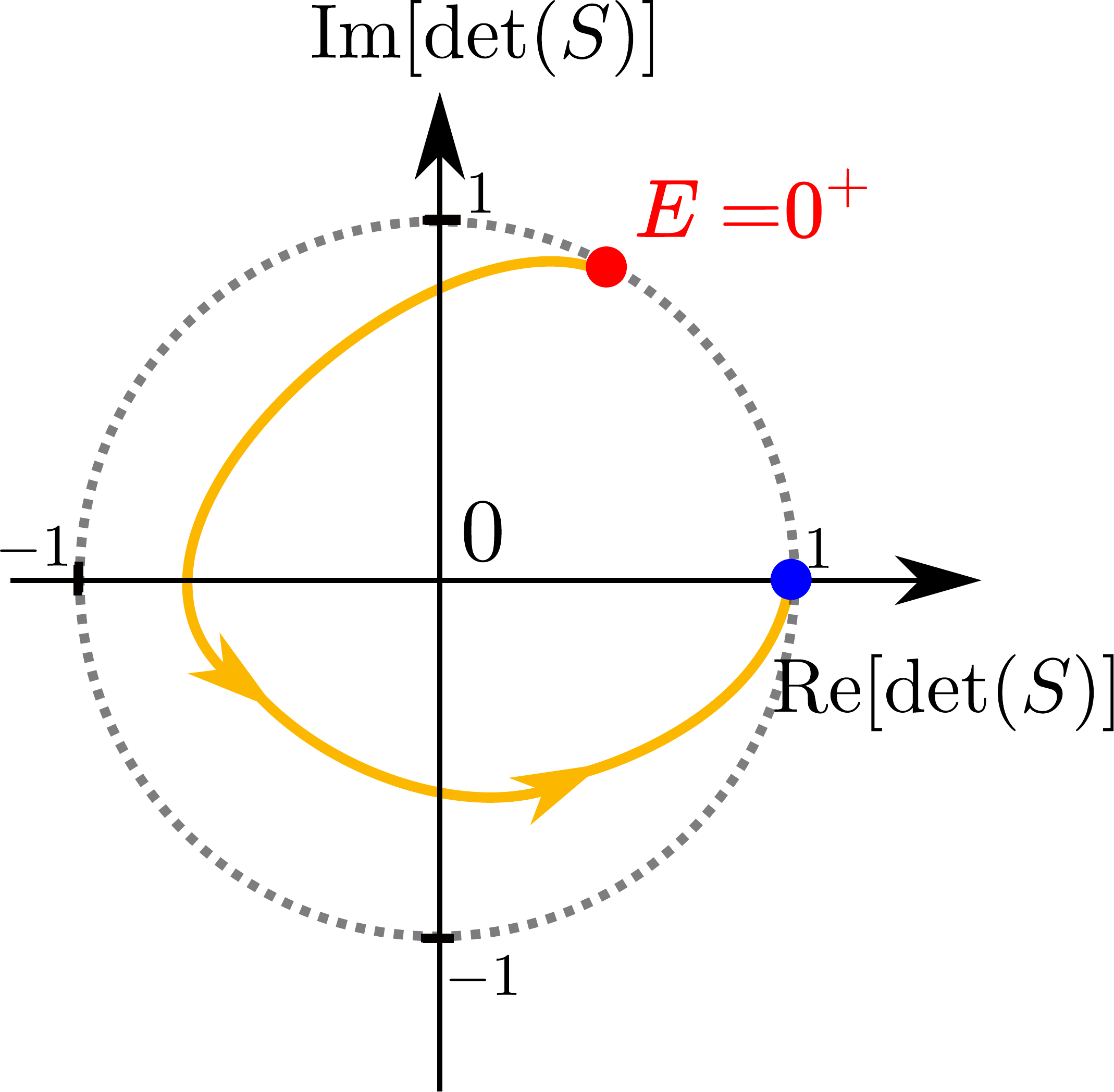}   \label{subfig:evendetS}}
     \subfloat[$\epsilon=\pm |d|k^5$]{
    \includegraphics[width=0.4\linewidth]{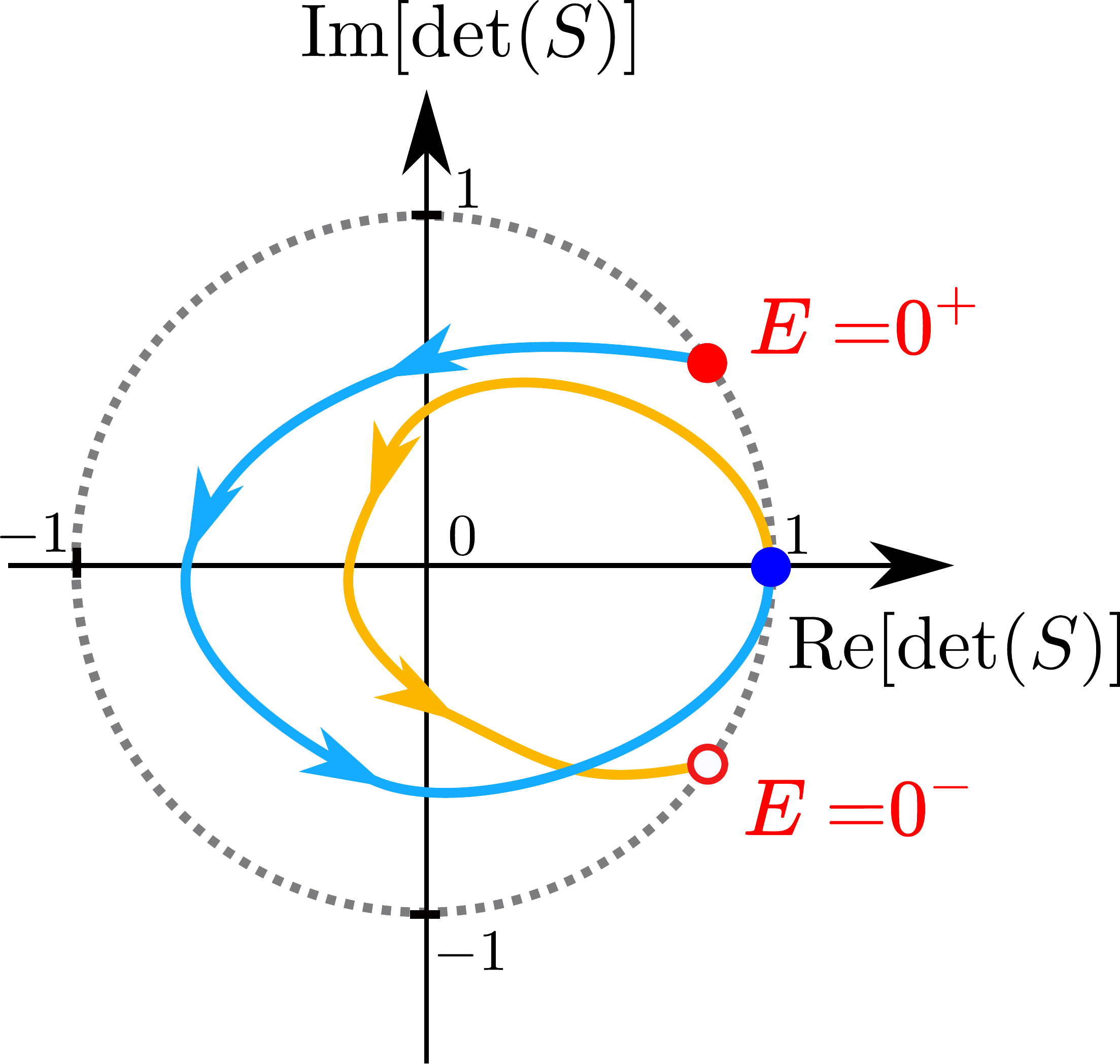}     \label{subfig:odddetS}}
        \caption{ Illustrations of the trajectories of $\det[\boldsymbol{S}(E)]$ of a dissipative system in the complex plane when $E$ is increased from $E_{min}$ to $E_{max}$ for (a) $\epsilon(k)=\pm|d| k$, where the trajectory starts and ends at $1$. (b) $\epsilon=|d|k^2$, where the trajectory starts at $-1$ and ends at $1$.
        (c) $\epsilon=|d| k^6$, where the trajectory starts at $\exp(i \pi/3)$ and ends at $1$.   
        (d) $\epsilon=\pm |d|k^5$, where the trajectory for $E\in (-\infty, 0)$ (solid yellow) starts at $1$ and ends at $S(0^-)=\exp(i\pi/5)$ , while the trajectory for $E\in (0, +\infty)$ (solid black) starts at $S(0^+)=\exp(i\pi/5)$ and ends at $1$.
    }
      \label{fig:jumps}
\end{figure} 
  For simplicity, we assume that there are no bright zero-energy eigenstates and no bound states in the continuum in the system.
Before discussing general $m$, we summarize the theorem for quadratic ($m=2$) and linear ($m=1$) dispersion relations. 
When energy $E$ is increased from the lower end of the continuum spectrum $E_{min}$ (which can be $-\infty$) to the upper end $E_{max}$ (which can be $+\infty$), $\det[\boldsymbol{S}(E)]$ traces a trajectory in the complex plane. 
 In the case of  
 $\epsilon(k)=k$,  the $S$-matrix is an identity matrix at both ends of the continuum spectrum. 
 The trajectory of $\det[\boldsymbol{S}(E)]$ in these cases forms a closed loop starting and ending at $1$, as illustrated in Fig.\ \subref*{subfig:lineardetS}. 
 For illustration purposes, we assume that the system is dissipative, so the trajectory is not confined to the unit circle. 
 Levinson's theorem states that the winding number of this loop around the origin is equal to the decrease in the number of bound states $\Delta N_B $ after the interaction is turned on \cite{ida1959relation, atkinson1966levinson}.
 For {\color{black}emitter} scattering,  the number of bound states for the bare Hamiltonian $H_0$ is equal to the number of {\color{black}emitters} $N$; hence, $\Delta N_B=N-N_B$, where  $N_B$ is the number of bound states for the full Hamiltonian \cite{wang2018single}.
 If we define the scattering phase $\delta(E)$ of $\det[\boldsymbol{S}(E)]\equiv |\det[\boldsymbol{S}(E)]|\exp(2i\delta(E))$ as a continuous function of $E$ \footnote{ For dissipative systems, we assume that $\det[\boldsymbol{S}(E)]\neq 0$ for any $E$ within the continuum spectrum}, the theorem can be stated as  $ \Delta \delta\equiv \delta(E_{max})-\delta(E_{min})=\pi\Delta N_B$. 
  For a quadratic dispersion relation $\epsilon(k)=k^2$, the trajectory of $\det[\boldsymbol{S}(E)]$ starts at  $\lim_{E\rightarrow 0}\det[\boldsymbol{S}(E)]=-1$ and ends at $\lim_{E\rightarrow +\infty}\det[\boldsymbol{S}(E)]=1$, as illustrated in Fig.\ \subref*{subfig:quadraticdetS}. 
  As compared to the closed-loop case of Fig.\ \subref*{subfig:lineardetS}, Levinson's theorem is modified to $ \Delta \delta=\pi \Delta N_B+\pi/2$ \footnote{For potential scattering, see Levinson's theorem for quadratic dispersion relation in 1D in Ref.\  \cite{dong2000levinson}. For emitter scattering, Levinson's theorem for quadratic dispersion relation is, to our knowledge, first presented in this Letter}.

Next, we give our  results  on Levinson's theorem for {\color{black}emitter} scattering with dispersion relation $\epsilon(k)=\sigma|d|k^m$ with $\sigma=\pm 1$ and photon-{\color{black}emitter}  couplings $\ket{v_k}=V(k)\ket{u}$.  First, consider the case of even $m$. 
When $\sigma=+1$, the trajectory of $\det[\boldsymbol{S}(E)]$ starts at $\lim_{E\rightarrow 0^+}\det[\boldsymbol{S}(E)]= \exp(2\pi i/m)$ [see Eq.\ \eqref{eqSeven}] and ends at $\lim_{E\rightarrow +\infty}\det[\boldsymbol{S}(E)]= 1$, as illustrated in Fig.\ \subref*{subfig:evendetS} for $m=6$. 
When $\sigma=-1$, the trajectory of $\det[\boldsymbol{S}(E)]$ starts at $\lim_{E\rightarrow -\infty}\det[\boldsymbol{S}(E)]= 1$ and ends at  $\lim_{E\rightarrow 0^-}\det[\boldsymbol{S}(E)]= \exp(-2\pi i/m)$.
In the Supplemental Material \cite{supp}, we prove that, for both cases,  
  \eq{
 \Delta \delta=\pi  (N-N_B)+\pi \frac{m-1}{m}\label{eqLev}.
  }
  When $m$ is odd, the continuum spectrum is $(-\infty,0)\cup (0,+\infty)$, and the trajectory of $S(E)$ is discontinuous across $0$, as illustrated in Fig.\ \subref*{subfig:odddetS}. 
  When $E$ increases from $-\infty$ to $0$, the trajectory starts from $1$ and ends at $\exp(-i\pi/m)$  {\color{black}[see Eq.\ \eqref{eqSodd}]}.  
  When $E$ increases from $0$ to $+\infty$, the trajectory starts at $\exp(+i\pi/m)$ and ends at $1$.  
  If we define $\Delta \delta$ as the sum of the winding phases of the two continuous trajectories, $\Delta \delta$ satisfies Eq.\ \eqref{eqLev}, as we show in the Supplemental Material \cite{supp}.

 \emph{Outlook.}--- In this Letter, we have illustrated how a divergent density of states results in a wide variety of universal scattering behaviors.  
{\color{black} An immediate next step is to generalize our results to arbitrary photon-emitter interactions and  non-separable short-range potentials.} 
Although our results rigorously apply only in the zero-energy limit, our work establishes the foundation for the development of a universal low-energy theory for {\color{black}general} dispersion relations. Similar to the case of quadratic dispersion relations, we expect the scattering to be primarily determined by the scattering length when the de Broglie wavelengths of the particles are large compared to the range of the interaction.  It will be interesting to explore how other well-studied problems for massive particles---such as Efimov physics  \cite{efimov1973energy, braaten2006universality, braaten2001three}, renormalization for the effective field theory \cite{adhikari1995perturbative, bedaque1999renormalization}, and the $N$-body scale  \cite{bazak2019four}---are modified in the presence of these more general dispersion relations. 
  
{ \color{black}  Our work also motivates new directions in many-body physics. The fact that bosons with quadratic dispersion relations form a Tonks–Girardeau gas at low-temperature in 1D and a Bose-Einstein condensate in 3D is closely related to the different behaviors of two-body scattering at the scattering threshold (total reflection vs. no interaction).  Our discovery of new nontrivial universal behaviors of the S-matrix may lead to predictions of new phases of dilute gases for systems with a divergent density of states.  Furthermore, it remains an outstanding challenge to describe emitter scattering when both dissipation and coherent driving are present.  }

\begin{acknowledgements}
We thank Sarang Gopalakrishnan, Zhen Bi, Darrick Chang, Abhinav Deshpande, {\color{black}Chris Baldwin}, and Simon Lieu for discussions.  
Y.W.\ and A.V.G.\ acknowledge support by  ARO MURI, AFOSR, AFOSR MURI, U.S.~Department of Energy Award No.~DE-SC0019449, DoE ASCR Quantum Testbed Pathfinder program (award No.~DE-SC0019040),  DoE ASCR Accelerated Research in Quantum Computing program (award No. DE-SC0020312), and NSF PFCQC program.
\end{acknowledgements}

\bibliographystyle{apsrev4-1}
\bibliography{library}

\begin{thebibliography}{50}
\bibitem [{con()}]{cont}%
  \BibitemOpen
  \href@noop {} {}\bibinfo {note} {The continuum spectrum for odd $m$ is given
  by $(-\infty,0)\cup (0,+\infty)$; for even $m$ is given by $(0,+\infty)$ when
  $\bold{s}=+1$ and $(-\infty, 0)$ when $\bold{s}=-1$.}\BibitemShut {Stop}%
\end{thebibliography}

\widetext
\clearpage
\begin{center}
\textbf{\large Supplemental Material}
\end{center}
\setcounter{equation}{0}
\setcounter{figure}{0}
\setcounter{table}{0}
\setcounter{page}{1}
\makeatletter
\renewcommand{\theequation}{S\arabic{equation}}
\renewcommand{\thefigure}{S\arabic{figure}}


\section{Overview}
 In this Supplemental Material, we present details omitted in the main text and generalize the results to higher dimensions, certain non-analytic dispersion relations and delta-function potential scattering.  In Sec.\ \ref{sec_ref2}, we derive the expression for $L(\omega)$ in Eq.\ (13) in the main text.  In Sec.\ \ref{secproofPs}, we prove Eq.\ (11) in the main text for the case of $N\geq 2$ and define bright zero-energy eigenstates. In Sec.\ \ref{SecLevinson}, we  prove Levinson's theorem [Eqs.\ (18) in the main text].  {\color{black} In Sec.\ \ref{SSecGen}, we generalize our results to higher dimensions and  dispersion relations $\epsilon(k)=|\bm{k}|^a$ with non-integer values of $a>0$. In Sec.\ \ref{SSecDelta}, we generalize our results to separable potential scattering.}

\renewcommand{\theequation}{S\arabic{equation}}
\renewcommand{\thefigure}{S\arabic{figure}}
\setcounter{equation}{0}
\setcounter{figure}{0}
\setcounter{secnumdepth}{1}

\section{Calculation of $L(\omega)$\label{sec_ref2} }

 In this section, we derive the expression for $L(\omega)$ in Eq.\ (13) in the main text. 
 We start with the definition of $L(\omega)$ in Eq.\ (9)  in the main text:
 \eq{
L(\omega)=\int_{-\infty}^{+\infty} dk \frac{1}{\omega-\epsilon(k)}. \label{eqSLdefi}
}
 The dispersion relation is given by $\epsilon(k)=\sigma |d|k^m$, where  $\sigma=\pm 1$ and $m\geq 2$ is a positive integer. To compute the integral, we close  the integration contour in the upper half \footnote{Closing the contour in the lower half of the complex plane would give the same answer } of the complex plane and apply the residue theorem:
 \eq{
 L(\omega)=-2\pi i \sum_{\text{Im}[y_j]>0}\epsilon'(y_j)^{-1}, \label{eqLsum}
 }
 where the complex numbers $y_i$ satisfy $\epsilon(y_j)=\sigma|d|y_i^m=\omega$ and $\text{Im}[y_j]>0$. 
Given the parametrization of $\omega$ in polar coordinates as $\omega=\sigma\exp(i\theta)|d|p^m $,
we have $y_j=\exp(i\theta/m)p \mu^{2j}$, where $\mu=\exp(i\pi/m)$ and $j\in \{0,1,\dots, m-1\}$. Define $A$ as the set of $j$ for which $y_j$ is above the real line.  Equation \eqref{eqLsum} can then be expressed as 
 \eq{
L(\omega)&=-2\pi i  \sum_{j\in A}\frac{1}{m|d| y_j^{m-1}},\label{eqKsum}\\
&=-\pi i   \frac{2}{m|d| p^{m-1} }\exp\left(-i\theta\frac{m-1}{m}\right)\sum_{j\in A} (-\mu)^{2j}, \label{eqKsum2}
}
where the set $A$ and the value of $\kappa_m\equiv \sum_{j\in A} (-\mu)^{2j} $ are given in Table \ref{table_value} for both odd and even $m$. 
Note that the prefactor $\frac{2}{m|d|p^{m-1}}$  in Eq.\ \eqref{eqKsum2} is equal to the density of states  $\rho(|\omega|)= \frac{2}{m|d|^{1/m}}|\omega|^{-1+1/m}$.   Hence, we have proved that $L(\omega)$ is given by Eq.\ (13) in the main text. 

\begin{table}[hb]
\caption{The set $A$ and the value of $\kappa_m = \sum_{j\in A} (-\mu)^{2j}$ for both odd and even $m$.}
\label{table_value}
\begin{tabular}{cl|c|c|}
\cline{3-4}
                                                &                     & $A$ & $\kappa_m\equiv \sum_{j\in A} (-\mu)^{2j}$ \\ \hline
\multicolumn{1}{|c|}{\multirow{2}{*}{Odd $m$}}  & $\theta\in (0,\pi)$ & $ (0, 1, 2,\dots \frac{m-1}{2}) $  & $-\frac{1}{\mu-1}$                \\ \cline{2-4} 
\multicolumn{1}{|c|}{}                          & $\theta\in (\pi, 2\pi)$              & $(0, 1, 2,\dots \frac{m-3}{2})$ & $-\frac{1}{(\mu-1)\mu}$                \\ \hline
\multicolumn{1}{|c|}{Even $m$}                          & $\theta\in (0,2\pi)$                  & $(0, 1, 2,\dots \frac{m-2}{2})$ & $\frac{2}{1-\mu^2}$            \\ \hline
\end{tabular}
\end{table}

\section{Emitter scattering
\label{secproofPs}}


  In this section, we prove that if there exists no bright zero-energy eigenstate, Eq.\ (11) in the main text holds for the class of models where $\ket{v_k}=V(k)\ket{u}$, even when $N\geq 2$. 
  Before diving into the proof, we give the definition of bright zero-energy eigenstates and give a physical explanation as to why our universality results require their absence. 
  
  
  Due to the multi-component nature of our emitter scattering problems, we find it necessary to  categorize all eigenstates of the Hamiltonian into bright, dark, and emitter eigenstates.  Bright eigenstates have a nonzero photon and emitter wavefunction, while dark eigenstates have only a nonzero photonic amplitude, and emitter eigenstates have only a nonzero emitter  amplitude. 
  With this terminology established, we now give an overview of the properties of the different types of eigenstates at zero energy. The zero-energy emitter states correspond to the null vectors of $\bm{K}^R$ that are orthogonal to $\ket{v_k}=V(k)\ket{u}$. They are decoupled from the photon channel, hence their existence has no impact on the universal behavior of the $S$-matrix. 
 For $V(k)$ with nonzero derivatives at $k=0$, there generally exist uncountably many zero-energy dark states independent of $\bm{K}^R$, which are polynomial functions with degree less than $m-1$. Bright states at zero-energy are fine-tuned and have a constant photon wavefunction in space.  As we show below, these states come into existence precisely when the universal scattering behavior fails.

  To give a heuristic explanation for why universal scattering at zero energy fails at these fine-tuned parameters, we consider the classic model of 1D potential scattering with quadratic dispersion relation ($m=2$), i.e., a 1D quantum mechanical problem described by the Schr{\" o}dinger equation
\begin{equation}
-\frac{d^2 \psi(z)}{d z^2} + V(z) \psi(z) = E \psi(z),
\end{equation}
where we set the mass equal to $1/2$.
A particle being scattered off a generic, short-range potential $V(z)$ would experience a total reflection in the limit $E\rightarrow 0$, similarly to what happens in our 1D emitter scattering models. Another feature of these 1D potential scattering problems is that there exists a fine-tuned, critical regime  when the scattering in the limit $E\rightarrow 0$ becomes total transmission instead of total reflection.  This occurs when there is a zero-energy eigenstate and there is no energy scale to compare with when the limit $E\rightarrow 0$ is taken. The zero-energy eigenstate can be understood as the effective ``transition state" when a new bound state emerges or disappears upon the continuous tuning of parameters. 

Similarly, in our emitter scattering models, the universal scattering behavior that takes place for generic parameters would fail at certain fine-tuned parameters.  An important difference to note is that, unlike in potential scattering, not all zero-energy eigenstates in emitter scattering are associated with the critical regime where the universal scattering behavior fails. For the particular type of interactions $\ket{v_k}$ being considered in this Letter,  we discover that the critical regime can be associated with the existence of a particular type of eigenstates at zero energy, which we call bright zero-energy states (defined above). 

  In order to state our goal more explicitly, we rewrite the Hamiltonian given by Eqs. (1)-(3) in the main text in the single-excitation manifold: 
  \eq{
H^{(1)}&=
\int_{-\infty}^{+\infty}dk\ \epsilon(k)C^\dagger(k)C(k)+\int_{-\infty}^{+\infty} dk \left[C^\dagger(k)  V^*(k)\bra{u}+C(k) V(k) \ket{u}\right]+\bm{K}^R,\label{EqEffHamil1}
}
 where we have used the matrix representation $\bm{K}^R$ to replace $\sum_{i,j=1}^N K^R_{ij}b_{i}^{\dagger}b_{j}$ and the vector $\ket{v_k}=V(k)\ket{u}$ to replace the emitter creation operators $\sum_{i=1}^N V_i(k)b^\dagger_i$.

  Our goal in this section is to prove the following theorem:

  \begin{theorem} \label{theorem1} Suppose $V(k)$ is a locally square-integrable function continuous at $k=0$ and $V(k=0)\neq 0$. When $k\rightarrow \pm\infty$, $|V(k)|^2=o(k^{m-\gamma})$ for some $\gamma>1$. Consider the class of emitter interactions $\ket{v_k}=V(k)\ket{u}$, where $\ket{u}$ is a unit vector. The single-particle $T$-matrix given by  Eqs.\ (6)- (8) in the main text reads:
\eq{
T(\omega, k,k')&=V^*(k')V(k)\langle u|\frac{1}{\omega\mathbb{1}_N-\boldsymbol{K}^R-\boldsymbol{K}(\omega)}|u\rangle,\label{eqSTG}\\
\boldsymbol{K}(\omega)&=|u\rangle \langle u| K(\omega),\quad K(\omega)\equiv\int_{-\infty}^{+\infty} dk \frac{|V(k)|^2}{\omega-\epsilon(k)}. \label{eqSKimp} 
}
When $H^{(1)}$ in Eq.\ \eqref{EqEffHamil1} has no bright zero-energy  eigenstates, Eq.\ (12) in the main text holds, namely,
   \eq{
\lim_{\omega\rightarrow 0}L(\omega)T(\omega,k,k')
= -\frac{V^*(k')V(k)}{|V(0)|^2}.\label{eqSLTeta0S}
}
Note that $T(\omega,k,k')$ and $\boldsymbol{K}(\omega)$ are defined for $\omega$  outside the continuum spectrum, hence the limit $\omega\rightarrow 0$ is taken in any direction except from within the continuum spectrum. 
  \end{theorem}

 \begin{proof}

Our proof consists of two lemmas linked by a condition on $\bm{K}^R$. The idea of the proof is that the absence of bright zero-energy  eigenstates can be translated into a condition on $\bm{K}^R$, which turns out to be necessary for the proof of Eq.\ \eqref{eqSLTeta0S}.

Choose an orthonormal basis $\{ |u_1\rangle,|u_2\rangle,\dots |u_{N}\rangle  \}$ for the single-emitter Hilbert space, where $\ket{u_1}\equiv \ket{u}$ is the first vector in this new basis. The link between the two lemmas  is the submatrix $\bm{K}^R_{\cancel{11}}$ constructed from deleting the first row and first column of $\bm{K}^R$; $\bm{K}^R_{\cancel{11}}$  can be considered as an operator on the emitter-excitation subspace $\{ |u_2\rangle,\dots |u_{N}\rangle  \}$ orthogonal to $\ket{u}$.
In Lemma \ref{Lemma9maintext}, we prove that Eq.\ \eqref{eqSLTeta0S} holds if any null vector of  $\bm{K}^R_{\cancel{11}}$ also corresponds to the null vector of $\bm{K}^R$. In Lemma \ref{lemmaBounds}, we prove that the condition  Lemma \ref{Lemma9maintext} relies on is guaranteed by the absence of bright  zero-energy eigenstates.  Combining the two lemmas  completes the proof of Theorem \ref{theorem1}. 

\begin{lemma}\label{Lemma9maintext}
If any null vector of  $\bm{K}^R_{\cancel{11}}$ also corresponds to the null vector of $\bm{K}^R$,
Eq.\ \eqref{eqSLTeta0S} follows.


\end{lemma}
\begin{proof}
Using Eq.\ \eqref{eqSTG}, the l.h.s of Eq.\ \eqref{eqSLTeta0S} can be written as
\eq{
\lim_{\omega\rightarrow 0} L(\omega)T(\omega,k, k')
& = V^*(k')V(k) \lim_{\omega\rightarrow 0} L(\omega) \bra{u}\bm{H}(\omega)^{-1}\ket{u},  \label{eqSLTstart}
}
where $\bm{H}(\omega)  \equiv  \omega\mathbb{1}_N-\boldsymbol{K}^R-\bm{K}(\omega)$.
Hence, our goal, Eq.\ \eqref{eqSLTeta0S}, is equivalent to 
\eq{
 \lim_{\omega\rightarrow 0} L(\omega) \bra{u}\bm{H}(\omega)^{-1}\ket{u}=-\frac{1}{|V(0)|^2}. \label{eqSLuHugoal}
}

In the new basis  where $|u_1\rangle= \ket{u}$ is the first basis vector, $\bra{u}\bm{H}(\omega)^{-1}\ket{u}$ is the $(1,1)$ matrix element of the inverse of $\bm{H}(\omega)$, and can be computed from the $(N-1)\times (N-1)$ submatrix $\bm{H}_{\cancel{11}}(\omega)$ constructed from deleting the first row and first column of $\bm{H}(\omega)$: 
\eq{
\bra{u}\bm{H}(\omega)^{-1}\ket{u}=\frac{\det(\bm{H}_{\cancel{11}}(\omega))}{\det(\bm{H}(\omega))}. \label{eqSsandH}
}
Using $\boldsymbol{K}(\omega)=|u\rangle \langle u| K(\omega)$, we have
\eq{
\det(\bm{H}_{\cancel{11}}(\omega))&=\det(\omega\mathbb{1}_{N-1}-\boldsymbol{K}^R_{\cancel{11}}),\label{eqSdetH11}\\
\det(\bm{H}(\omega))&
=-K(\omega)\det(\omega\mathbb{1}_{N-1}-\boldsymbol{K}^R_{\cancel{11}})+\det(\omega\mathbb{1}_N-\boldsymbol{K}^R).\label{eqSdetH}
}
Combining Eqs.\ \eqref{eqSsandH}, \eqref{eqSdetH11} and \eqref{eqSdetH}, the l.h.s of Eq.\ \eqref{eqSLuHugoal} becomes 
\eq{
\lim_{\omega\rightarrow 0} L(\omega) \bra{u}\bm{H}(\omega)^{-1}\ket{u}
&=\lim_{\omega\rightarrow 0} L(\omega)\left(-K(\omega)+\frac{\det(\omega\mathbb{1}_N-\boldsymbol{K}^R)}{\det(\omega\mathbb{1}_{N-1}-\boldsymbol{K}^R_{\cancel{11}})}\right)^{-1}. \label{eqSLimLuHu}
}
Let us label the $N$ roots of the characteristic polynomial of    $\bm{K}^R$ by $E_i$ for $i=1,\dots N$, and the $N-1$ roots of the characteristic polynomial of  $\bm{K}_{\cancel{11}}^R$ by $\bar{E}_i$ for $i=1,\dots N-1$. $E_i$ and $\bar{E}_i$ correspond to the eigenvalues of $\bm{K}^R$ and $\bm{K}_{\cancel{11}}^R$, respectively, where any eigenvalue with multiplicity $n\geq 2$ is assigned to $n$ different indices. We have
\eq{
\lim_{\omega\rightarrow 0}\frac{\det(\omega\mathbb{1}_N-\boldsymbol{K}^R)}{\det(\omega\mathbb{1}_{N-1}-\boldsymbol{K}^R_{\cancel{11}})}=\lim_{\omega\rightarrow 0}\frac{\prod_{i=1}^N (\omega-E_i)}{\prod_{i=1}^{N-1} (\omega-\bar{E}_i)}. \label{eqSdetratioK}
}
Since any null vector of $\bm{K}^R_{\cancel{11}}$  corresponds to a null vector of $\bm{K}^R$ by the assumption of the Lemma, if $\bm{K}^R$ has null vectors, its zero-eigenvalue multiplicity must be greater or equal to that of $\bm{K}_{\cancel{11}}$. Hence, the limit in Eq.\ \eqref{eqSdetratioK} is finite.

In the main text, we have introduced the identity $\lim_{\omega\rightarrow 0} L^{-1}(\omega)\frac{1}{\omega-\epsilon(k)}=\delta(k)$; hence $\lim_{\omega\rightarrow 0} L^{-1}(\omega)K(\omega)=|V(0)|^2\neq 0$ and Eq.\ \eqref{eqSLimLuHu} leads to Eq.\ \eqref{eqSLuHugoal}. The proof of Lemma  \ref{Lemma9maintext} is complete.

\end{proof}
If we can prove that the absence of bright  zero-energy  eigenstates of Eq.\ \eqref{EqEffHamil1} guarantees that any null vector of  $\bm{K}^R_{\cancel{11}}$ also corresponds to the null vector of $\bm{K}^R$, Eq.\ \eqref{eqSLTeta0S} would immediately follow from Lemma \ref{Lemma9maintext}. To do this, we prove the contrapositive statement in the following lemma:
\begin{lemma}
\label{lemmaBounds}
When there exists a vector $\ket{e_0}=\sum_{i=2}^N e_i\ket{u_i}$  orthogonal to $\ket{u}$, such that $\bm{K}^R\ket{e_0}\neq 0$ and  $\bm{K}_{\cancel{11}}^R\ket{e_0}=0$, then there exists a bright zero-energy eigenstate of the Hamiltonian in Eq.\ \eqref{EqEffHamil1}. 
\end{lemma}
\begin{proof}
We plan to write down an ansatz with a nonzero photon and emitter wavefunction and verify that it is a zero-energy eigenstate of the Hamiltonian in Eq.\ \eqref{EqEffHamil1}. 
 The ansatz we propose is the following: 
\eq{
\ket{\psi_0}&=\int_{-\infty}^{+\infty} dz\  \psi_0(z)C^\dagger(z)\ket{0,g}
+\ket{e_0},\label{eqSPsi0ansatz}\\ 
\psi_0(z)&=-V(0)^{-1}\braket{u|\bm{K}^R|e_0},\label{eqSPsi0ansatzPh}
}
where the photon wavefunction $\psi_0(z)$ in the coordinate space is a constant function.
By definition, $\ket{e_0}$ is orthogonal to $\ket{u}$. Because $\bm{K}^R\ket{e_0}\neq 0$ and  $\bm{K}_{\cancel{11}}^R\ket{e_0}=0$,  $\bm{K}^R\ket{e_0}$ is a nonzero vector proportional to $\ket{u}$. Hence, $\psi_0(z)\neq 0$. 

Our goal is to prove that the ansatz given by Eqs.\ \eqref{eqSPsi0ansatz} and \eqref{eqSPsi0ansatzPh} is the zero-energy eigenstate of the Hamiltonian in Eq.\ \eqref{EqEffHamil1}.
Applying  $H^{(1)}$ in Eq.\ \eqref{EqEffHamil1} to the  Fourier transform of the ansatz in Eq.\ \eqref{eqSPsi0ansatz}, we get 
\eq{
H^{(1)}\ket{\psi_0}&=\int_{-\infty}^{+\infty} dk\ \left[ \epsilon(k) \psi_0(k)C^\dagger(k)\ket{0,g} +V^*(k)\braket{u|e_0} C^\dagger(k)\ket{0,g} + \psi_0(k)V(k)\ket{u}\right]+\bm{K}^R\ket{e_0}, \label{eqSHpsi0}
}
where the momentum-space photon wavefunction $\psi_0(k)= -V(0)^{-1}\braket{u|\bm{K}^R|e_0}\delta(k)$ is the Fourier transform of Eq.\ \eqref{eqSPsi0ansatzPh}.  Since the dispersion relation satisfies $\epsilon(0)=0$, the first term on the r.h.s of  Eq.\ \eqref{eqSHpsi0} is zero: $\int_{-\infty}^{+\infty} dk\  \epsilon(k) \psi_0(k)C^\dagger(k)\ket{0,g}=0$.
Because $\braket{u|e_0}=0$, the second term on the r.h.s of  Eq.\ \eqref{eqSHpsi0} is also equal to $0$. The third term
\eq{
\int_{-\infty}^{+\infty} dk\  \psi_0(k)V(k)\ket{u}=-\ket{u}\bra{u}\bm{K}^R\ket{e_0}
}
cancels with the fourth term $\bm{K}^R\ket{e_0}$ on the r.h.s of  Eq.\ \eqref{eqSHpsi0} because $\bm{K}^R\ket{e_0}$ is proportional to $\ket{u}$. Therefore, $H^{(1)}\ket{\psi_0}=0$, and this is the end of the proof for Lemma  \ref{lemmaBounds}.  

\end{proof}
Combining Lemmas \ref{Lemma9maintext} and \ref{lemmaBounds}, we can obtain Theorem \ref{theorem1}.

\end{proof}
\section{Levinson's theorem \label{SecLevinson}}  
In this section, we prove Levinson's theorem for the class of emitter scattering models with $\ket{v_k}=V(k)\ket{u}$, i.e. Eq.\ (18) in the main text. 
Let us restate the objective of our proof in the following theorem:
\begin{theorem} \label{theoremLevinsonEmitter} 
   Denote the continuum spectrum by $\mathcal{R}_c$. We assume that there are no bound states in the continuum or bright zero-energy eigenstates in the system.  For dissipative systems, we assume that $\det[\boldsymbol{S}(E)] \neq 0$ for $E\in \mathcal{R}_c$. The winding phase $\Delta \delta$ of $\det[\boldsymbol{S}(E)]$ around the origin is defined as  
\eq{
 2\Delta \delta&=-i \int_{\mathcal{R}_c} dE\frac{\partial_E\det[\boldsymbol{S}(E)]}{\det[\boldsymbol{S}(E)]}. \label{eqDeltaPhi}
  }

 Suppose $\ket{v_k}=V(k)\ket{u}$ satisfies the properties listed in Theorem \ref{theorem1} and the  dispersion relation is given by $\epsilon(k)=\sigma|d|k^m$, where $\sigma=\pm 1$ and $m\geq 2$ is an integer.  We have
\eq{
\Delta \delta=\pi (N-N_B)+\pi \frac{m-1}{m} \label{eqThLevImp},
}
where $N$ is the number of emitters and $N_B$ is the number of bound states.


\end{theorem}

The main idea of the proof is to define an analytic continuation of $\det[\boldsymbol{S}(E)]$ to the complex plane and observe the fact that the bound state energies are the poles of this function. The proof is similar to our previous work \cite{wang2018single}, where we proved Levinson's theorem for photon-emitter models with linear dispersion relations.

In preparation for the proof of Theorem \ref{theoremLevinsonEmitter},  we introduce  Theorem \ref{TheoremDetS}, where we propose an analytic function $s(\omega)$ that is equal to the analytic continuation of $\det[\boldsymbol{S}(E)]$ to the complex plane.  Though introduced here as a tool for proving Theorem \ref{theoremLevinsonEmitter}, Theorem \ref{TheoremDetS} provides a quick method to compute $\det[\boldsymbol{S}(E)]$ using $\bm{K}^{R}$ and $\bm{K}(E+i0^\pm)$  and is an important theorem itself. We comment that the range of application of Theorem \ref{TheoremDetS} is well beyond the class of photon-emitter models discussed in this letter: it can be applied to general photon-emitter interactions $\ket{v_k}$ and other dispersion relations beyond $\epsilon(k)=\pm |d|k^m$. 

\begin{theorem}\label{TheoremDetS}

Define $J(\omega)=\det[\omega \mathbb{1}_N-\boldsymbol{K}^R-\boldsymbol{K}(\omega)]$ as a function on the complement of the continuum spectrum $\mathcal{R}_c$ in the complex plane.  For the values of $\omega$ s.t. $J(\omega)\neq 0$, we can define $s(\omega)=\frac{J(\omega^*)}{J(\omega)}$. 
When $E$ is not equal to the energy of a bound state in the continuum, 
 \eq{
s(E+i0^+) =\det[\boldsymbol{S}(E)]. \label{eqSSS} 
 }
 \end{theorem}
 We comment that the bound state energies $E_B$ correspond to the poles of the emitter propagator $\bm{G}(\omega)=[\omega \mathbb{1}_N-\boldsymbol{K}^R-\boldsymbol{K}(\omega)]^{-1}$, hence they  satisfy $J(E_B)=0$. 
 \begin{proof} 
 Let $n(E)$ denote the momentum degeneracy at energy $E$ and $k_1, \dots k_{n(E)}$ the degenerate momenta at energy $E$.
 When $\epsilon(k)=\pm |d| k^m$ and $E\in \mathcal{R}_c$, $n(E)=1$ for odd $m$ and $n(E)=2$ for even $m$. 
According to Eqs.\ (5)-(8) in the main text, the $S$-matrix $\bm{S}(E)$ is a $n(E)\times n(E)$ matrix, whose matrix elements are given by 
 \eq{
S_{\alpha\beta}(E)=\delta_{\alpha\beta}- \frac{2\pi i}{\sqrt{|\epsilon'(k_\alpha(E))\epsilon'(k_\beta(E))|}}\braket{v_{k_\beta(E)}|\boldsymbol{G}(E+i0^+)|v_{k_\alpha(E)}},\label{eqSgvTs}
}
where $\alpha, \beta \in (1,2, \dots, n(E))$. 

Note that in writing down Eq.\ \eqref{eqSgvTs}, we have implicitly assumed that the limit $\bm{G}(E+i0^+)\equiv \lim_{\eta\rightarrow 0^+}\bm{G}(E+i\eta)$ exists.
However, if $E_c \mathbb{1}_N-\boldsymbol{K}^R-\boldsymbol{K}(E_c+ i0^+)$ has a zero eigenvalue for some energy $E_c\in \mathcal{R}_c$,  $\bm{G}(E_c+ i\eta)$ does not have a limit when $\eta\rightarrow 0^+$ and $E_c$ corresponds to the energy of a bound state in the continuum. This is why the theorem only applies to  $E\neq E_c$. 



Construct $\bm{A}$ as a $N\times n(E)$ matrix and $\bm{A}^\dagger$ its Hermitian conjugate: 
\eq{
 \bm{A}&=\left(\begin{matrix}
\frac{1}{\sqrt{{|\epsilon'(k_1)|}}} |v_{k_1}\rangle, \dots
\frac{1}{\sqrt{{|\epsilon'(k_{n(E)})|}}}|v_{k_{n(E)}}\rangle
\end{matrix}\right), \quad 
} then the $n(E)\times n(E)$ matrix 
$\bm{S}(E)$ for $E\neq E_c$ can be written as 
\eq{
\bm{S}(E)=\mathbb{1}_{n(E)}-2\pi i \bm{A}^\dagger \boldsymbol{G}(E+i0^+) \bm{A}\label{eqSnTnexp},
} 
where $\mathbb{1}_{n(E)}$ is an identity matrix of dimension $n(E)$. 

Using the definitions of $s(\omega), J(\omega)$ and the properties of determinant, we have, 
  \eq{
s(E+i0^+)
&=\det( \mathbb{1}_N+(\boldsymbol{K}(E+i0^+)-\boldsymbol{K}(E+i0^-))\boldsymbol{G}(E+i0^+)),\label{eqSsEp0}
 }
where $ \boldsymbol{K}(E+i0^+)-\boldsymbol{K}(E+i0^-)$ can be re-written as
  \eq{
 \boldsymbol{K}(E+i0^+)-\boldsymbol{K}(E+i0^-)&=\int_{-\infty}^{+\infty}dk\ |v_k\rangle\langle v_k|\left(\frac{1}{E+i0^+-\epsilon(k)}-\frac{1}{E-i0^--\epsilon(k)}\right),\nonumber\\
 &=-\int_{-\infty}^{+\infty}dk\ |v_k\rangle\langle v_k|2\pi i  \delta(E-\epsilon(k)), \nonumber\\
 &=-2\pi i \sum_{\alpha=1}^{n(E)} \frac{1}{|\epsilon'(k_\alpha)|}|v_{k_\alpha}\rangle \langle v_{k_\alpha}|=-2\pi i \bm{A}\bm{A}^\dagger. \label{eqSKpKm}
 }
Inserting Eq.\ \eqref{eqSKpKm} into Eq.\ \eqref{eqSsEp0}, we get
 \eq{
s(E+i0^+)=\det\left[\mathbb{1}_N-2\pi i \bm{A}\bm{A}^\dagger\boldsymbol{G}(E+i0^+)\right].\label{eqdetratioexp}
}
According to a standard result in linear algebra known as the extension of the matrix determinant lemma, given an invertible $N\times N$ matrix $-2\pi i\boldsymbol{G}(E+i0)$ and a  $N\times n(E)$ matrix $\bm{A}$, 
\eq{
\det[\mathbb{1}_N-2\pi i\bm{A}\bm{A}^\dagger\boldsymbol{G}(E+i0^+)]=\det[\mathbb{1}_{n(E)}-2\pi i\bm{A}^\dagger\boldsymbol{G}(E+i0^+)\bm{A}] \label{eqSdoubledet}.
}
Using Eqs.\ \eqref{eqSnTnexp} and \eqref{eqdetratioexp}, we see that the l.h.s and r.h.s of Eq.\ \eqref{eqSdoubledet} are equal to $s(E+i0^+)$ and $\det[\boldsymbol{S}(E)]$, respectively. This is the end of the proof for Theorem \ref{TheoremDetS}.

 \end{proof}

We proceed to prove Theorem \ref{theoremLevinsonEmitter} with the help of Theorem  \ref{TheoremDetS}. 
\begin{proof} (Theorem \ref{theoremLevinsonEmitter})

 First consider  $\epsilon(k)=|d|k^m$ with odd $m$, in which case the continuum spectrum is  $\mathcal{R}_c=(-\infty,0)\cup (0,+\infty)$. Since we have assumed that there is no bound state in the continuum,  using Theorem\ \ref{TheoremDetS}, we can replace $\det[\boldsymbol{S}(E)]$  by $\frac{J(E+i0^-)}{J(E+i0^+)}$ for $E\in \mathcal{R}_c$ and  rewrite Eq.\ \eqref{eqDeltaPhi} in terms of a contour integral in the complex plane:
\eq{
 2\Delta \delta & = -i \int_{\mathcal{R}_c} d E \bigg[ \frac{\partial_{E} J(E+ i 0^-)}{J(E + i0^-)} - 
 \frac{\partial_{E} J(E + i 0^+)}{J(E + i0^+)} \bigg] =
 -i\int_{\mathcal{R}_1^++\mathcal{R}_1^-+\mathcal{R}_2^++\mathcal{R}_2^-}d\omega \frac{\partial_{\omega} J(\omega)}{J(\omega)}, \label{eqfirstline}
}
where $E$ is a real coordinate, $\omega$ is a complex coordinate, and the integration contours $\mathcal{R}_1^{\pm}$ and $\mathcal{R}_2^{\pm}$ are illustrated by the dashed lines in Fig.\ \ref{subfig:cubic_eye}.

\begin{figure}
    \centering
    
    \subfloat[$\epsilon(k)=|d| k^m$, odd $m$]{
    \includegraphics[width=0.33\linewidth]{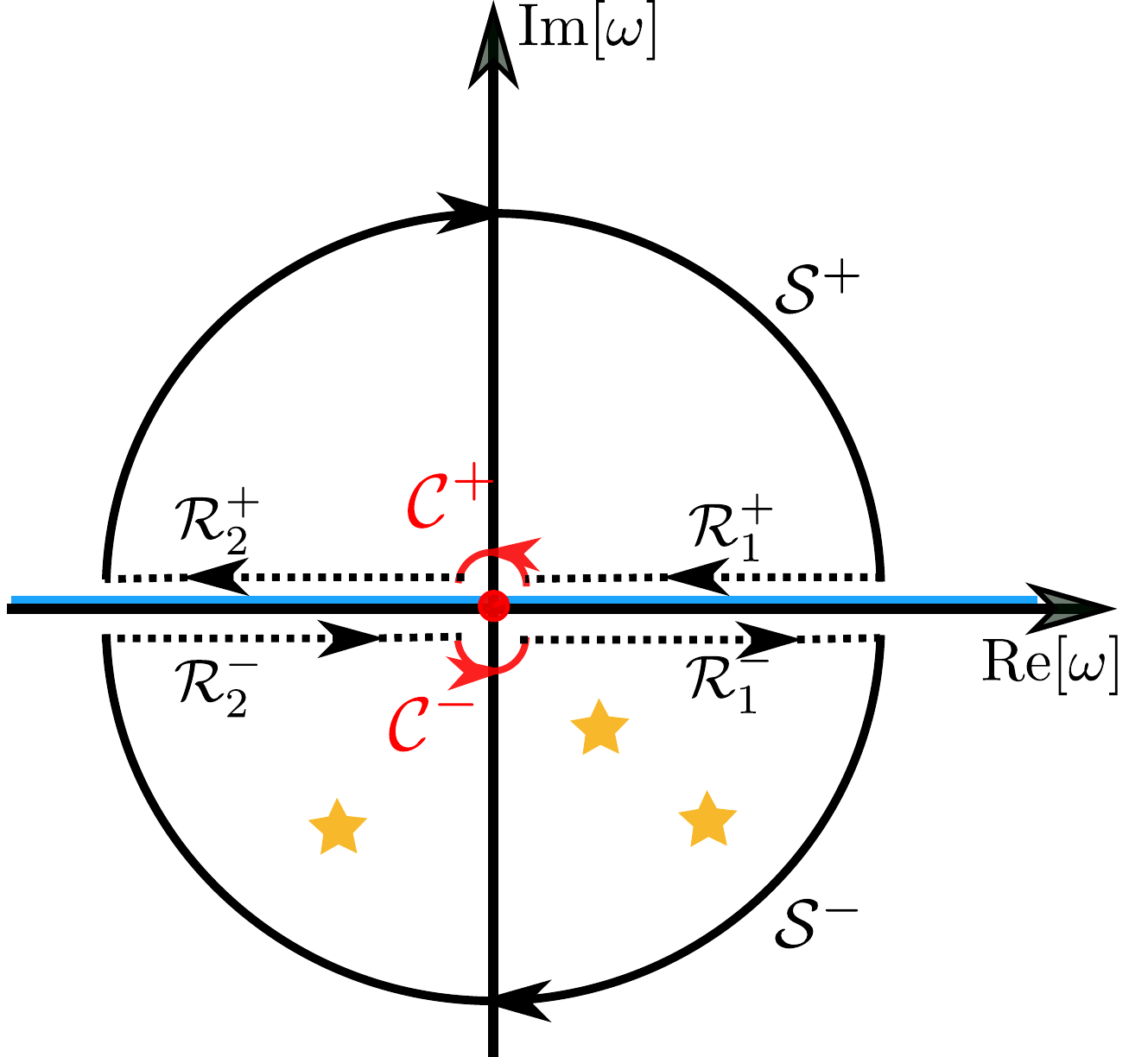}   \label{subfig:cubic_eye}
}
     \subfloat[$\epsilon(k)=|d|k^m$, even $m$]{
    \includegraphics[width=0.33\linewidth]{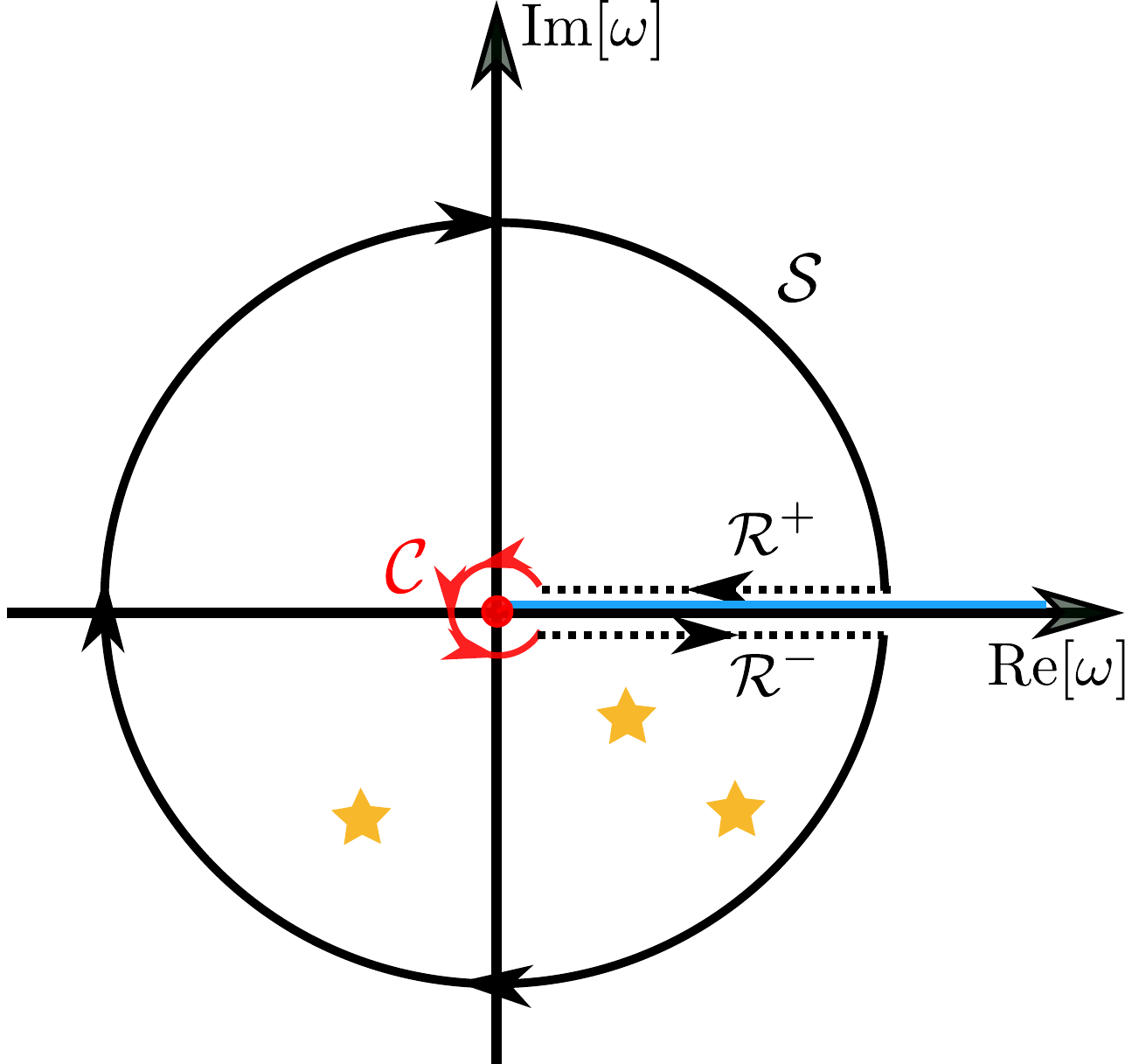}  \label{subfig:quadratic_eye} 
}

      \caption{Illustration of the integration contours for the calculation of the winding phase of $\det[\boldsymbol{S}(E)]$.   
      (a) Contours for a dispersion relation $\epsilon(k)= |d|k^m$ with odd $m$.
      (b) Contours for a dispersion relation $\epsilon(k)= |d|k^m$ with even $m$. 
      The density of states diverges at the origin $E=0$, marked by the red dot. 
      The black lines represent the continuum spectrum,  while the yellow stars represent bound-state energies. 
      The dashed lines with arrows are the integration paths for the evaluation of the  winding number of $\det[\boldsymbol{S}(E)]$. 
      The semicircles (circles) are added to form closed integration contours so that the residue theorem can be invoked. 
      The red semicircles (circle) go around the origin with an infinitesimal radius, while the black semicircles (circle) have an infinite radius.}
 \end{figure}
$\mathcal{R}_2^{\pm}$ and $\mathcal{R}_1^{\pm}$ represent the contours just above/below the real line for $E<0$ and $E>0$, respectively.  
We can obtain two closed integration contours by adding a pair of semicircles $\mathcal{C}^{\pm}$ with an infinitesimal radius around $0$ and a pair of semicircles $\mathcal{S}^{\pm}$ with radius $|\omega|\rightarrow \infty$. 
Equation \eqref{eqfirstline} can then be rewritten as
 \eq{
2 \Delta \delta &=-i\oint d\omega\frac{\partial_{\omega} J(\omega)}{J(\omega)} +i \int_{\mathcal{S}^++\mathcal{S}^-} d\omega\frac{\partial_{\omega} J(\omega)}{J(\omega)} +i\int_{\mathcal{C}^++\mathcal{C}^-}  d\omega\frac{\partial_{\omega} J(\omega)}{J(\omega)},\label{eqWinding}
 }
 where $\oint$ represents the sum of integrals over the two closed contours. For odd $m$, $J(\omega)$ is analytic in the complement of the real line in the complex plane.   The poles of $J^{-1}(\omega)$ correspond to the bound state energies;  they can only be located below the real line given our assumption that there is no bound state in the continuum. This also implies that  when $\bm{K}^R$ is Hermitian, $N_B=0$.  Applying the residue theorem, the closed contours in the upper and lower half planes yield $0$ and $-2\pi i N_B$, respectively. Hence,
 \eq{
 -i\oint d\omega\frac{\partial_{\omega} J(\omega)}{J(\omega)}=-2\pi  N_B. \label{eqSLSclosed}
 }

 Next, we evaluate the integrals along the small semicircles. $J(\omega)$ is equal to $\det[\bm{H}(\omega)]$ in Eq.\ \eqref{eqSdetH}, which shows that  $J(\omega)\sim K(\omega)\sim L(\omega)$ when $\omega\rightarrow 0$.  Intuitively, the winding phases of $J(\omega)$ along $\mathcal{C}^\pm$ are equal to the winding phases of $L(\omega)$ along $\mathcal{C}^\pm$, which contribute to the term $\pi \frac{m-1}{m}$ in Eq.\ \eqref{eqThLevImp}. To demonstrate it rigorously, we  write $J(\omega)$ as the product of $L(\omega)$ and another function $g(\omega)$:
 \eq{
J(\omega)&= L(\omega)g(\omega),\nonumber\\
g(\omega)&\equiv 
-L^{-1}(\omega)K(\omega)\det(\omega\mathbb{1}_{N-1}-\boldsymbol{K}^R_{\cancel{11}})+L^{-1}(\omega)\det(\omega\mathbb{1}_N-\boldsymbol{K}^R). \label{eqSdefig}
}
This way the winding phases of $J(\omega)$ along $\mathcal{C}^\pm$ can be evaluated as the sum of the winding phases of  $g(\omega)$ and $L(\omega)$:
 \eq{
i\int_{\mathcal{C}^\pm}  d\omega\frac{\partial_{\omega} J(\omega)}{J(\omega)}
&=i\int_{\mathcal{C}^\pm}  d\omega\frac{\partial_{\omega} L(\omega)}{L(\omega)}
+i\int_{\mathcal{C}^\pm}  d\omega\frac{\partial_{\omega} g(\omega)}{g(\omega)}. 
  }
   Using  Eqs.\ \eqref{eqKsum2}, the winding number of $L(\omega)$ can be evaluated explicitly in polar coordinates:
 \eq{
i\int_{\mathcal{C}^+}  d\omega\frac{\partial_{\omega} L(\omega)}{L(\omega)}&= i\lim_{r\rightarrow 0} \int_{0^+}^{\pi^-}d\theta\ \frac{\partial_\theta L(r,\theta)}{L(r,\theta)}, \nonumber\\
&=i\lim_{r\rightarrow 0} \int_{0^+}^{\pi^-}d\theta\ - i(m-1)/m, \nonumber\\
&=\pi (m-1)/m, \label{eqSintCp0}
 }
 where $\int_{0^+}^{\pi^-}d\theta \equiv \lim_{\theta_1\rightarrow 0^+,\theta_2\rightarrow \pi^-}\int_{\theta_1}^{\theta_2} d\theta$. The integral along  $\mathcal{C}^-$ can be evaluated similarly; and it has the same value as the integral along $\mathcal{C}^+$.
 

Next, we argue that the winding phases of $g(\omega)$  along $\mathcal{C}^\pm$ are equal  to $0$. Note that the contour $\mathcal{C}^+/\mathcal{C}^-$ is defined through two limiting processes taken consecutively on an arc centered at the origin of the complex plane.  In the first limit, we fix the radius of the arc and send both endpoints of the arc to infinitesimal distances above/below the real line, so the arc almost becomes a semicircle. In the second limit, the radius of the arc is sent to $0$. 
 Because of this, we need to  first examine  $g(E+ i\eta)$ when $\eta \rightarrow 0^\pm$, and then send $E\rightarrow 0$. 

Using Eq.\ \eqref{eqSdefig}, we see that $g(\omega)$ is an analytic function in the complement of the real line on the complex plane for odd $m$. Since $\lim_{\eta \rightarrow 0^\pm}L^{-1}(E+i\eta)K(E+i\eta)$ and $\lim_{\eta \rightarrow 0^\pm}L^{-1}(E+i\eta)$ exist for $E$ anywhere on the real line $\mathcal{R}$,  $\lim_{\eta\rightarrow 0^\pm}g(E+i\eta ) \equiv g(E+i0^\pm)$ exist  for $E\in \mathcal{R}$. Furthermore, since $\lim_{\omega \rightarrow 0}L^{-1}(\omega)K(\omega)=|V(0)|^2$,  $g(E+ i0^\pm)$ as functions of $E\in \mathcal{R}$ are  continuous at $E=0$.

The winding phase of $g(\omega)$ along $ \mathcal{C}^+$ is equal to the phase difference between $g(-|E|+i0^+)$ and $g(|E|+i0^+)$ 
in the limit $E\rightarrow 0$.    Similarly, the winding phase of $g(\omega)$ along $ \mathcal{C}^-$ is equal to the phase difference between $g(|E|+i0^-)$ and $g(-|E|+i0^-)$ 
  in the limit $E\rightarrow 0$.
Because of the continuity of $g(E+ i0^\pm)$ at  $E=0$, $i\int_{\mathcal{C}^++\mathcal{C}^-}  d\omega\frac{\partial_{\omega} g(\omega)}{g(\omega)}=0$.
Therefore, 
\eq{
i\int_{\mathcal{C}^++\mathcal{C}^-}  d\omega\frac{\partial_{\omega} J(\omega)}{J(\omega)}=i\int_{\mathcal{C}^++\mathcal{C}^-}  d\omega\frac{\partial_{\omega} L(\omega)}{L(\omega)}=2\pi (m-1)/m.\label{eqSintCp}
}

 At last, we evaluate the integral along the large semicircles, which can be written in polar coordinates as 
\eq{
i\int_{\mathcal{S}^+ +\mathcal{S}^-}  d\omega
\frac{\partial_{\omega} J(\omega)}{J(\omega)}=i\lim_{r\rightarrow \infty}\left(\int_{\pi^-}^{0^+}  d\theta \frac{\partial_\theta  J(r,\theta)}{J(r,\theta)}
+\int_{2\pi^-}^{\pi^+}  d\theta \frac{\partial_\theta  J(r,\theta)}{J(r,\theta)}\right).\label{eqSintegralSpmPolar}
}
Using Eq.\ \eqref{eqSdetH},  $J(\omega)$
can be written as 
\eq{
J(\omega)=-K(\omega)P_{N-1}(\omega) +P_{N}(\omega), \label{eqSJP}
}
where $P_{N-1}(\omega)\equiv\det(\omega\mathbb{1}_{N-1}-\boldsymbol{K}^R_{\cancel{11}})$ and $P_N(\omega)\equiv\det(\omega\mathbb{1}_N-\boldsymbol{K}^R)$ are polynomial functions of $\omega$ with degrees $N-1$ and $N$, respectively.
From the definition of $K(\omega)$ in Eq.\ \eqref{eqSKimp}, we can see that $\lim_{|\omega|\rightarrow \infty}K(\omega)=0$, hence  $J(\omega)\sim P_N(\omega)\sim \omega^N$ when $|\omega|\rightarrow \infty$; and we expect that the sum of the winding phases of $J(\omega)$ around the large semicircles is equal to $2\pi N$. In the following, we provide a careful mathematical analysis to verify this intuitive result. 

Taking the derivative of Eq.\ \eqref{eqSJP} w.r.t  $\theta$, we have
\eq{
\partial_{\theta}J(r,\theta)=i\omega \partial_{\omega}J(\omega)=-i\omega\partial_\omega K(\omega)P_{N-1}(\omega) -i\omega K(\omega)\partial_{\omega}P_{N-1}(\omega) +i\omega\partial_\omega P_{N}(\omega). \label{eqSdthetaJ}
} 
We can observe from Eq.\ \eqref{eqSKimp} that $ \lim_{|\omega|\rightarrow \infty}\partial_\omega K(\omega)=\lim_{|\omega|\rightarrow \infty} K(\omega)=0$. In addition, $\lim_{|\omega|\rightarrow \infty}\frac{P_{N-1}(\omega)}{P_N(\omega)}=\lim_{|\omega|\rightarrow \infty}\frac{\partial_{\omega}P_{N-1}(\omega)}{P_N(\omega)}=0$, 
hence
\eq{
\lim_{r\rightarrow \infty}\frac{\partial_{\theta}J(r,\theta)}{J(r,\theta)}=\lim_{|\omega|\rightarrow \infty}\frac{ i\omega\partial_\omega P_{N}(\omega)}{P_{N}(\omega)}=iN \label{eqLimrinfpartialJ}
}
 uniformly in $\theta\in (0,\pi)\cup (\pi,2\pi)$. Applying the dominated convergence theorem,  we can evaluate the $r\rightarrow \infty$ limit  of the following definite integral as a function of the integration end points $\theta_1, \theta_2\in (0,\pi)$ (or  $\theta_1, \theta_2\in (\pi, 2\pi)$):
\eq{
i\lim_{r\rightarrow \infty}\int_{\theta_1}^{\theta_2} d\theta\ \frac{\partial_\theta  J(r,\theta)}{J(r,\theta)}=(\theta_1-\theta_2) N. \label{eqlimthetarinftyJ}
}
The limit in Eq.\ \eqref{eqlimthetarinftyJ} is  uniform in  $\theta_1, \theta_2$ because the limit in Eq.\ \eqref{eqLimrinfpartialJ} is uniform in $\theta$. This implies that, when we evaluate Eq.\ \eqref{eqSintegralSpmPolar}, we can exchange the limit in $r$ and the limits in the integration end points:
\eq{
i\int_{\mathcal{S}^+ }  d\omega
\frac{\partial_{\omega} J(\omega)}{J(\omega)} & \equiv 
 i \lim_{r\rightarrow 0}\lim_{\theta_1\rightarrow \pi^-}\lim_{ \theta_2\rightarrow 0^+} \int_{\theta_1}^{\theta_2} d\theta\ \frac{\partial_\theta  J(r,\theta)}{J(r,\theta)}\nonumber\\
&=i \lim_{\theta_1\rightarrow \pi^-}\lim_{ \theta_2\rightarrow 0^+} \lim_{r\rightarrow 0}\int_{\theta_1}^{\theta_2} d\theta\ \frac{\partial_\theta  J(r,\theta)}{J(r,\theta)}=\pi N, \label{eqSintSp}
}
where we have used Eq.\ \eqref{eqlimthetarinftyJ} in evaluating Eq.\ \eqref{eqSintSp}. 
The integration along $\mathcal{S}^-$ can be evaluated similarly, and it has the same value as Eq.\ \eqref{eqSintSp}; hence we get 
\eq{
i\int_{\mathcal{S}^+ +\mathcal{S}^-} d\omega \frac{\partial_{\omega} J(\omega)}{J(\omega)}=2\pi  N .\label{eqSintegralSpmPolar2}
}
Combining Eqs. \eqref{eqWinding}, \eqref{eqSLSclosed}, \eqref{eqSintCp} and \eqref{eqSintegralSpmPolar2}, we obtain Eq.\ \eqref{eqThLevImp} for the dispersion relation $\epsilon(k)=|d|k^m$ with odd $m$.
The case of $\epsilon(k)=-|d|k^m$ with odd $m$ can be proved identically once we replace $E$ with $-E$. 

Next, we discuss the case of $\epsilon(k)=|d|k^m$ with even $m$. 
Similarly as in the case of odd $m$, the winding phase of $\det[\boldsymbol{S}(E)]$ can be evaluated as
\eq{
2 \Delta \delta & =\int_{0^+}^{\infty} d E \bigg[ \frac{\partial_{\omega} J(E- i 0)}{J(E - i0)} - 
 \frac{\partial_{\omega} J(E + i 0)}{J(E + i0)} \bigg] =\int_{\mathcal{R}^++\mathcal{R}^-}d\omega \frac{\partial_{\omega} J(\omega)}{J(\omega)} \nonumber\\
 &=\oint d\omega\frac{\partial_{\omega} J(\omega)}{J(\omega)} - \int_{\mathcal{S}} d\omega\frac{\partial_{\omega} J(\omega)}{J(\omega)} -\int_{\mathcal{C}}  d\omega\frac{\partial_{\omega} J(\omega)}{J(\omega)},
 }
 where the integration contours are illustrated in Fig.\ \ref{subfig:quadratic_eye}. 
 $\mathcal{R}^\pm$ represent the contours just above and below the real line along the continuum spectrum.  
 $\oint$ represents the integration over the  closed contour.
Following a procedure similar to the case of odd $m$, it is easy to show that the result of this integral is also given by Eq.\ \eqref{eqThLevImp}. The case of $\epsilon(k)=-|d|k^m$ with even $m$ can be proved similarly. 
 
This is the end of the proof for Theorem \ref{theoremLevinsonEmitter}.
\end{proof}

{\color{black}

\section{Generalization to spatial dimension $D$ and non-integer $m$ \label{SSecGen} }
\subsection{Angular momentum eigenstates in $D$ dimension}
In the main text and  Sec.\ \ref{SSecDelta} of the Supplement,  we have focused on 1D systems with dispersion relations $\epsilon(k)=\pm|d|k^m$, where $m$ is a positive integer. To demonstrate the generality of the principle that divergent density of states leads to a nontrivial universal limit of the $S$-matrix, we extend the discussion to all dimensions $D\geq 1$ and non-integer values of $m$. 
Let $\bm{k}$ denote the momentum vector in integer spatial dimension $D\geq 1$. For simplicity, we assume a dispersion relation with rotational symmetry: $\epsilon(k)=|\bm{k}|^a$, where $a>0$ does not have to be an integer.  These dispersion relations are natural extensions of the even integer $m$ case of $\epsilon(k)= \sigma |d| k^m$ in the one-dimensional models.  The odd integer extensions of this analytic dispersion relation do not have natural analogs for $D>1$.
The density of states is defined as 
\eq{
\rho(E)&=\int  d^D k\ \delta(E-|\bm{k}|^a)=b(D)\int_0^{+\infty} dk\ k^{D-1} \delta(E-k^a)=b(D)k^{D-1}\epsilon'(k)^{-1}, \label{eqSrhoEdefinition}
}
where the constant $b(D)=\frac{2\pi^{D/2}}{\Gamma(D/2) } $ comes from the integration over the solid angle of a $(D-1)$-sphere. $\Gamma(z)$ is the gamma function. Evaluating Eq.\ \eqref{eqSrhoEdefinition}, we have 
\eq{
\rho(E)=\frac{b(D)}{a}E^{-1+\zeta^{-1}}=\frac{b(D)}{a} k^{D-a}, \quad \zeta\equiv\frac{a}{D}. \label{eqSrhohigher}
}
where $k\equiv E^{1/a}$.
When $D=1$, we have $b(D)=2$ and $\zeta=a$, so that Eq.\ \eqref{eqSrhohigher} agrees with the value of $\rho(E)$ in the main text for 
$a = m$ for positive even integer $m$. For general values of $D$,  $\rho(E)$ diverges when $m> D$; we will show that the $S$-matrix goes to a nontrivial limit dependent on $a/D$ at zero energy. $\rho(E)$  has a finite limit at $E=0$ when $a\leq D$; we will show that the $S$-matrix goes to the identity matrix at zero energy.  

Let us first study the $S$-matrix in $D$ dimensions. The momentum-space representation of the scattering operator 
is given by
\eq{
\mathcal{S}(\bm{k},\bm{k'})&=\delta(\bm{k}\!-\!\bm{k}')-2\pi i \delta[\epsilon(\bm{k})\!-\!\epsilon(\bm{k'})] T(E\!+\!i0^+, \bm{k},\bm{k'}),\label{SeqST}
}
where $T(E\!+\!i0^+, \bm{k},\bm{k'})$ is the momentum-space representation of the $T$-operator $\mathcal{T}(\omega)$ which we specify later. 
 Energy is preserved in the scattering process, hence we can define an operator $\bm{S}(E)$ that describes the scattering process at energy $E$. In 1D, the momentum degeneracy is $2$ at any energy; $\bm{S}(E)$ is a $2\times 2$ matrix the same as in the case of positive even integer $m$ discussed in the main text. In 3D and higher dimensions,  there are uncountably many momentum eigenstates at the same energy $E>0$; $\bm{S}(E)$ is an integral operator in the momentum basis instead of a discrete, finite matrix. 
 
 In the familiar cases of  quadratic dispersion relation $a=2$ and $D=2,3$,  it is a common practice to choose the common eigenstates of the angular momentum operator and the kinetic energy operator as the basis states for the representation of the $S$-matrix. For example, when $D=3$, the angular momentum eigenbasis can be labelled by two integers, $l$ and $m_l$, where
$l = 0,1,2,\dots $ is called the angular momentum quantum number and $m_l=-l,-l+1,\dots, l$ the magnetic quantum number. In the angular momentum eigenbasis, the scattering operator at energy $E$ can be represented as a matrix, describing the transmission coefficients between different angular momentum eigenstates at energy $E$.   Note that the dispersion relation $|\bm{k}|^a$ shares the same eigenbasis as the quadratic dispersion relation $|\bm{k}|^2$, hence we can use the same angular momentum eigenbasis for the representation of the $S$-matrix.  
In the following, we give an overview of the angular momentum eigenstates in arbitrary dimensions.


In  $D\geq 2$ dimensional space  with Cartesian coordinates $\{z_1,\dots z_{D}\}$, we can introduce generalized polar coordinates $\{ r, \theta_1,\dots \theta_{D-1}\}$ such that $r=\sum_{i=1}^{D} z_i^2$ is the radial distance to the origin of the coordinate frame.  The set $\bm{\theta}=\{\theta_1,\dots \theta_{D-1}\}$  specifies coordinates on the surface of a $(D-1)$-sphere  \cite{louck1960theory, granzow1963n}.  
The $D$-dimensional total orbital angular-momentum operator is given by  $L^2\equiv - \nabla_{\bm{\theta}}^2$, where $\nabla_{\bm{\theta}}^2$ is the  Laplacian operator on the unit $(D-1)$-sphere---a partial differential operator defined purely in terms of  $\bm{\theta}$. The eigenvalues and eigenvectors of $L^2$ are given by 
\eq{
L^2Y_{l,q_l}(\bm{\theta})=l(l+D-2)Y_{l,q_l}(\bm{\theta}),
}
where $l=0, 1, 2,\dots$ is the generalization of the angular momentum quantum number to $D$ dimensions and $q_l=1,2,\dots N_l$ labels the degenerate eigenstates.  $Y_{l,q_l}(\bm{\theta})$ is the generalization of spherical harmonics to $D$ dimensions  \cite{efthimiou2014spherical}.  When $l=0$, $N_l=1$, i.e.~the eigenstate is non-degenerate.   When $l\geq 1$, $N_l=\frac{D+2l-2}{l}C^{D+l-3}_{l-1}$.  For example, when $D=2$, $N_l=2$ for $l\geq 1$; $Y_{l,q_l=1}(\theta)=(2\pi)^{-1/2}\exp(il\theta)$ and $Y_{l,q_l=2}(\theta)=(2\pi)^{-1/2}\exp(-il\theta)$, where $\theta=\arctan(z_2/z_1)$.  When $D=3$, $N_l=2l+1$ for $l\geq 1$; $q_l$ has a one-to-one correspondence with the magnetic quantum number $m_l=-l,-l+1,\dots l $.

The orthogonality relations of the spherical harmonics are given by 
\eq{
\int d\Omega \ Y^*_{l,q_l}(\bm{\theta})Y_{n,q_n}(\bm{\theta})=\delta_{nl} \delta_{q_l, q_n},
}
where $\int d\Omega $ is the integration over the solid angle of the $(D-1)$-sphere. 

  In scattering theory with $D\geq 2$, states with $l=0,1,2\dots $ are often referred to as s-waves, p-waves, d-waves..., respectively. When $D=1$, the dispersion relation is symmetric about $k=0$; the s-wave and  p-wave refer to the symmetric and antisymmetric combinations of the degenerate momentum eigenstates at a given energy, respectively.  In the main-text discussion of 1D systems, we have shown that  the scattering of the s-wave is decoupled from the p-wave when $E\rightarrow 0$; the s-wave transmission coefficient has a nontrivial limit $\exp(i\pi/a)$, while the p-wave transmission coefficient is $1$. The goal of this section is to generalize the zero-energy scattering behavior in 1D  to higher dimensions. Specifically, in systems with nonvanishing interactions at zero energy, s-wave scattering is decoupled from all other channels in the zero-energy limit in any dimension; the s-wave transmission coefficient goes to a universal limit  $\exp(2\pi i D/a)$ when $a>D$, while the scattering in other channels  $l\geq 1$ goes to a trivial limit---the identity matrix. 
  
  The different zero-energy behaviors for $l=0$ and $l\geq 1$ are due to the different behaviors of the radial wavefunctions $R_l(r)\sim (kr)^l$ at small $r$. The main idea is that,  when $E\rightarrow 0$, $R_l(r)$ goes to a constant at any finite $r$ for $l=0$ and vanishes for $l\geq 1$. Therefore, the s-wave experiences interactions at zero energy while the higher channels do not see the interactions. To substantiate the argument, let us compute $R_l(r)$ below. 


The eigenvalue equation for the kinetic energy operator at positive energy $E$ is given by 
\eq{
 \nabla^{2} \phi(\bm{r}) =k^2\phi(\bm{r}), \label{eqSDLaplaOrigin}
}
where $k\equiv E^{1/a}$ and $ \nabla^{2} $ is the Laplacian operator in $D$ dimension. In the spherical coordinate system $(r, \bm{\theta})$, we have
\eq{
 \nabla^{2}=r^{1-D}\frac{\partial}{\partial r}\left(r^{D-1}\frac{\partial}{\partial r}\right)-\frac{L^2}{r^2}. 
 \label{eqSlaplaceAngular}
}
Inserting the separable ansatz $\phi_{l,q_l,E}(\bm{r})=R_{l,k}(r)Y_{l,q_l}(\bm{\theta})$ into Eq.\ \eqref{eqSDLaplaOrigin} and using Eq.\ \eqref{eqSlaplaceAngular}, we obtain 
the radial equation 
\eq{
\left[\frac{d^2}{dr^2}+\frac{D-1}{r}\frac{d}{dr}-\frac{l(l+D-2)}{r^2}+k^2\right]R_{l,k}(r)=0. \label{eqSradialD}
}
Defining $R_{l,k}(r)=r^{-\frac{D-1}{2}}y(r)$, Eq.\ \eqref{eqSradialD} can be written as 
\eq{
\left[\frac{d^2}{dr^2}-\frac{l'(l'+1)}{r^2}+k^2\right]y(r)=0, \quad l'\equiv l+\frac{D-3}{2}, \label{eqSradialDy}
}
where $l'\geq 0$ for $D\geq 2$. When $D=1$,  the centrifugal term $\frac{l'(l'+1)}{2}$ vanishes and $R_l(r)=y(r)$;  when $D=3$, we have $l'=l$, $R_l(r)=y(r)/r$.
When $D\geq 2$, as $r\rightarrow 0$, the centrifugal term $\frac{l'(l'+1)}{2}$ dominates the energy term $k^2$, and the solutions behave like solutions of the corresponding equation with $E=0$; namely, like combinations of $r^{l'+1}$ and $r^{-l'}$. Thus, the physically acceptable wave function behaves like $r^{l'+1}$--- the Riccati-Bessel function $\hat{j}_{l'}(kr)$\cite{ARFKEN2013643}:
\eq{
\hat{j}_{l'}(z)&\equiv z j_{l'}(z)\equiv \sqrt{\frac{\pi z}{2}}J_{l'+\frac{1}{2}}(z)\\
&=z^{l'+1}\sqrt{\frac{\pi}{2 }}\sum_{n=0}^{\infty}\frac{(-z^2/2)^n}{n!2^{n+l'+1/2}\Gamma(l'+n+3/2)} \label{eqSriccati}
}
where $j_{l'}(z)$ is the spherical Bessel function, and $J_{\lambda}(z)$ the ordinary Bessel function. 
The Riccati-Bessel functions satisfy the following orthogonality relations
\eq{
\int_0^{\infty}dr \hat{j}_{l'}(k'r)\hat{j}_{l'}(kr)=\frac{\pi}{2}\delta(k-k'). \label{eqSorthoBessel}
}
Hence we obtain $R_{l,k}(r)$ and $\phi_{l,q_l,E}(r,\bm{\theta})$:
\eq{
R_{l,k}(r)&= \sqrt{\frac{2}{\pi }} r^{-\frac{D-1}{2}} \hat{j}_{l'}(kr), \label{eqRlk}\\
\phi_{l,q_l,E=k^a}(r,\bm{\theta})&=[\rho(E)\epsilon'(k)]^{-1/2}R_{l,k}(r)Y_{l,q_l}(\bm{\theta}), \nonumber\\
&=\sqrt{\frac{2}{\pi }}b(D)^{-1/2}(kr)^{-\frac{D-1}{2}}\hat{j}_{l'}(kr)Y_{l,q_l}(\bm{\theta}).
 \label{eqPhilqlE}
 }
Here, we have chosen normalization constants such that the following orthogonality and completeness relations are satisfied:
\eq{
&\braket{\phi_{l,q_l,E}|\phi_{n,q_n,E'}}
=\rho(E)^{-1} \delta(E-E')\delta_{l,n}\delta_{q_l,q_n},\label{Seqorthogon}\\
&\int_0^{\infty} dE\ \sum_{l,q_l}  \rho(E)\,  \ket{\phi_{l,q_l,E}}\bra{\phi_{l,q_l,E}}=\mathbb{1}, \label{SeqComplete}
}
where $\mathbb{1}$ is the identity operator in the Hilbert space of a particle in $D$ dimensions.

Using Eqs.\ \eqref{eqSriccati}, \eqref{eqRlk} and \eqref{eqPhilqlE},  we have, for $kr\ll 1$, 
\eq{
\phi_{l,q_l,E}(r,\bm{\theta})
=&b(D)^{-1/2} Y_{l,q_l}(\bm{\theta}) \frac{(kr)^{l}}{2^{l+D/2-1}
\Gamma(l+\frac{D}{2})}\{1+O[(kr)^2]\}.
}
Hence, using 
$ Y_{0,1}(\bm{\theta})=b(D)^{-\frac{1}{2}}$, we can derive the point-wise convergence
\eq{
\lim_{E\rightarrow 0^+}\phi_{l,q_l,E}(\bm{r})=\begin{cases}
b(D)^{-1}\frac{1}{2^{D/2-1}\Gamma(\frac{D}{2})}=\left(\frac{1}{\sqrt{2\pi}}\right)^D & l=0\\
0 & l\geq 1
\end{cases}, \label{eqSwavefunctionE0}
}
which confirms our earlier claim that the s-wave has a constant wavefunction at zero energy.

Eq.\ \eqref{eqSwavefunctionE0} is all we need to know about the angular momentum wavefunctions to derive the $S$-matrix universal limits. Eq.\ \eqref{eqSwavefunctionE0} implies that the $S$-matrix universal limit is only nontrivial for $l=0$; the quantum number $q_l$ plays no role in the  discussions. For simplicity and uniformity of notation with the 1D case, we will use a single variable $\alpha=1,2,\dots $ to denote the pair of quantum numbers $(l,q_l)$. In particular, $\alpha=1$ corresponds to $l=0$ and $\alpha=2,3,\dots$ correspond to states with $l\geq 1$.  The orthogonality relation in Eq.\ \eqref{Seqorthogon} can be rewritten as
\eq{
\braket{\phi_{\alpha,E}|\phi_{\beta,E'}}=\rho(E)^{-1} \delta(E-E')\delta_{\alpha,\beta}.
}


Finally, we are ready to define the $S$-matrix in the angular momentum basis in arbitrary dimension. 
In the basis $\{\ket{\phi_{\alpha,E}}\}$, the scattering operator at energy $E$ can be represented by a matrix $\bm{S}(E)$:
\eq{
\rho(E)\braket{\phi_{\beta,E}|\mathcal{S}|\phi_{\alpha,E'}}&\equiv\delta(E-E')S_{\alpha, \beta}(E),\nonumber \\
S_{\alpha, \beta}(E)&=\delta_{\alpha,\beta}-2\pi i
\rho(E) T(E\!+\!i0^+, \alpha, E, \beta, E),\label{eqSgvT1}
}
where $S_{\alpha,\beta}(E)$ is the matrix element of $\bm{S}(E)$ and
\eq{
 T(\omega, \alpha, E, \beta, E)\equiv \bra{\phi_{\beta, E}} \mathcal{T}(\omega) \ket{\phi_{\alpha,E}}.
}
Eq.\ \eqref{eqSgvT} can be compared to Eq.\ (5) in the main text for 1D systems.

\subsection{Universal scattering}

In this subsection, we consider emitter scattering for arbitrary integer spatial dimension $D\geq 1$ and dispersion relation $\epsilon(k)=|\bm{k}|^a$, where $a>0$ is not required to be an integer.   The Hamiltonian is given by 
\eq{
H&=H_0+V,\nonumber\\
H_0&=\int d^D k \ \epsilon(\bm{k})C^\dagger(\bm{k})C(\bm{k})+\sum_{i,j=1}^N K^R_{ij}b_{i}^{\dagger}b_{j},\nonumber \\
V&=\int d^D k \left[\sum_{i=1}^N V_i(\bm{k})C(\bm{k})  b^{\dagger}_{i}+\text{h.c}\right],\label{EqEffHamilsup}  
}
where we have either commutation or anti-commutation relations: $[C(\bm{k}),C^\dagger(\bm{k}')]_{\pm}=\delta(\bm{k}-\bm{k}'), [b_i, b^\dagger_j]_{\pm}=\delta_{ij}$.  We assume that the emitter-photon interaction has the form $\ket{v_{\bm{k}}}\equiv (V_1(\bm{k}),\dots V_N(\bm{k}))^T=V(\bm{k})\ket{u}$, where $\ket{u}$ is a unit vector. Let $\bm{0}$ be the null vector in dimension $D$.  We require that $V(\bm{k})$ is locally square-integrable and continuous at $\bm{k}=\bm{0}$ and that $V(\bm{0})$  is nonzero.

Similarly to the main text, we can define a $N\times N$ matrix $\bm{K}(\omega)$ describing the effective interactions between emitters:
\eq{
K_{ij}(\omega)=\int d^D k\ \frac{ V_i(\bm{k})V_j^*(\bm{k})}{\omega-|\bm{k}|^a}. \label{eqSKijD}
}
The momentum-space representation of the $T$-operator $\mathcal{T}(\omega)$ is given by
\eq{
T(\omega, \bm{k}, \bm{k}')=\braket{v_{ \bm{k}'} |\frac{1}{\omega\mathbb{1}_N-\bm{K}(\omega)-\bm{K}^R}|v_{\bm{k}}}.
} 
Since $V(\bm{k})$ is square-integrable,  its Fourier transform  $\tilde{V}(\bm{z})=\left(\frac{1}{\sqrt{2\pi}}\right)^D\int d^D k \exp(i\bm{k}\bm{z}) V(\bm{k})$ exists.  To find the representation of the $T$-operator in the basis $\{\ket{\phi_{\alpha,E}}\}$,  define vector $|v_{\alpha, E}\rangle=V_{\alpha,E}\ket{u}$, where 
\eq{
V_{\alpha,E}&\equiv \int d^D z\ \tilde{V}(\bm{z})\phi^*_{\alpha,E}(\bm{z}). \label{eqSValphaDefi}
}
The vector elements of $|v_{\alpha, E}\rangle$ represent the interaction coefficients between the emitters and the angular momentum mode $\alpha$ at energy $E$. 

The representation of the $T$-operator in the angular momentum basis is given by
\eq{
T(\omega,\alpha, E, \beta, E')=V^*_{\beta,E'}V_{\alpha,E}\braket{u |\frac{1}{\omega\mathbb{1}_N-\bm{K}(\omega)-\bm{K}^R}|u}. \label{eqTomegaDvv}
}
The $S$-matrix in the angular momentum basis is related to $T(\omega,\alpha, E, \beta, E')$ by  Eq.\ \eqref{eqSgvT}. The generalization of  1D universal scattering to arbitrary dimension $D$ and to all (including non-integer) $m > 0$ is given in the following theorem: 
\begin{theorem} \label{theoremHighDuniS}
Suppose $V(\bm{k})$ is a locally square-integrable function continuous at $\bm{k}=0$ and suppose $V(\bm{k}=0)\neq 0$.
 In the absence of bright zero-energy eigenstates, when $a\leq D$,  $\lim_{E\rightarrow 0^+}S_{\alpha,\beta}(E)=\delta_{\alpha,\beta}$;  when $a> D$, 
 \eq{
\lim_{E\rightarrow 0^+}S_{\alpha,\beta}(E)=\begin{cases}
\exp(2\pi i D/a) & \alpha=\beta=1, \\
\delta_{\alpha,\beta} & \mathrm{Otherwise}. \label{eqSSDtheorem}
\end{cases}
}

\end{theorem}
 \begin{proof}

 In the orthonormal basis where $\ket{u}$ is the first basis vector, the only nonzero matrix element of $\bm{K}(\omega)$ is $K_{11}(\omega)\equiv K(\omega)$:
\eq{
K(\omega)&\equiv\int d^D k\ \frac{ V(\bm{k})V^*(\bm{k})}{\omega-|\bm{k}|^a},\\
\bm{K}(\omega)&=\ket{u}\bra{u}K(\omega).
}
Using Eq.\ \eqref{eqTomegaDvv}, we have
 \eq{
\lim_{\omega\rightarrow 0}T(\omega,\alpha, E, \beta, E')=\lim_{\omega\rightarrow 0} V^*_{\beta,E'}V_{\alpha,E}\left(-K(\omega)+\frac{\det(\omega\mathbb{1}_N-\boldsymbol{K}^R)}{\det(\omega\mathbb{1}_{N-1}-\boldsymbol{K}^R_{\cancel{11}})}\right)^{-1}. \label{eqSsandH3}
}
It is easy to show that Lemma \ref{lemmaBounds}  applies to general dispersion relations and any dimension. Hence, the absence of bright zero-energy eigenstates implies that any null vector of  $\bm{K}^R_{\cancel{11}}$ corresponds to the null vector of $\bm{K}^R$.
  Using a similar argument as in Eq.\ \eqref{eqSdetratioK},  $\lim_{\omega\rightarrow 0}\frac{\det(\omega\mathbb{1}_N-\boldsymbol{K}^R)}{\det(\omega\mathbb{1}_{N-1}-\boldsymbol{K}^R_{\cancel{11}})}$ exists in the absence of bright zero-energy eigenstates. 
 
We first prove the theorem for cases when the S-matrix has a trivial zero-energy limit. When $a<D$, $\lim_{E\rightarrow 0^+}\rho(E)=0$ from Eq.\ \eqref{eqSrhohigher};  $\lim_{\omega\rightarrow 0}K(\omega)$ is a constant. Hence, 
$\lim_{\omega\rightarrow 0}T(\omega,\alpha, E, \beta, E')$ exists from Eq.\ \eqref{eqSsandH3}. Using Eq.\ \eqref{eqSgvT}, we can conclude that $\lim_{E\rightarrow 0^+}S_{\alpha,\beta}(E)=\delta_{\alpha,\beta}$ when $a<D$.

 When $a=D$, $\lim_{E\rightarrow 0}\rho(E)$ is finite from Eq.\ \eqref{eqSrhohigher}; $K(\omega)$ diverges logarithmically in the limit of $\omega\rightarrow 0$. Hence, we have $\lim_{\omega\rightarrow 0}T(\omega,\alpha, E, \beta, E')=0$ from Eq.\ \eqref{eqSsandH3}. Using Eq.\ \eqref{eqSgvT}, we see that $\lim_{E\rightarrow 0^+}S_{\alpha,\beta}(E)=\delta_{\alpha,\beta}$ when $a=D$.

We continue with the proof of nontrivial universal limit of $\bm{S}(E)$ when $a>D$.  
Similarly, as in the 1D case, define 
\eq{
L(\omega)=\int d^D k \frac{ 1}{\omega-|\bm{k}|^a}=b(D) \int_0^{+\infty} dk\ \frac{k^{D-1} }{\omega-k^a} =\frac{b(D)}{D}\int_0^{+\infty} dp \frac{1}{\omega-p^\zeta},
}
where, in the last equality, $\zeta=\frac{a}{D}$ and we have changed the integration variable from $k$ to $p=k^{D}$.  
The integral converges for $\zeta>1$, and the value of $L(\omega)$ is given by 
\eq{
L(\omega)=-\pi i \rho(|\omega|)\frac{2}{1-\exp(2\pi i/\zeta)}\exp\left(-i\theta \frac{\zeta-1}{\zeta}\right), \label{SeqLzeta}
}
which diverges at the same rate as the density of states $\rho(|\omega|)$ when $|\omega|\rightarrow 0$. Eq.\ \eqref{SeqLzeta} agrees with Eq.\ (13) in the main text for even $m=a$ and $D=1$.

When $\zeta>1$, following a standard relation in functional analysis, $\lim_{\omega\rightarrow 0}L(\omega)\frac{1}{\omega-|\bm{k}|^a}=\delta(\bm{k})$. We have
 \eq{
\lim_{\omega\rightarrow 0} K(\omega)L^{-1}(\omega)=|V(\bm{0})|^2\neq 0. \label{eqSKLD}
 }
Using Eqs.\  \eqref{eqSsandH3} and \eqref{eqSKLD}, we have, in the absence of bright zero-energy eigenstates, 
\eq{
\lim_{\omega\rightarrow 0}L(\omega)T(\omega,\alpha, E, \beta, E')=
-\frac{V^*_{\beta,E'}V_{\alpha,E}}{|V(\bm{k}=\bm{0})|^2}, \label{eqSLTarbD}
    }
which can be compared to Eq.\ \eqref{eqSLTeta0S} for the case of $D=1$. 
Using Eqs.\ \eqref{eqSwavefunctionE0} and \eqref{eqSValphaDefi}, we have
\eq{
\lim_{E\rightarrow 0}V_{\alpha,E}=\begin{cases}\left(\frac{1}{\sqrt{2\pi}}\right)^D\int d^D z \ \tilde{V}(\bm{z})=V(\bm{k}=\bm{0}) & \alpha=1,\\
0 & \alpha\geq 2,
\end{cases} \label{eqSTalphabeta}
}
where $\alpha=1$ corresponds to  $l=0$ and $\alpha\geq 2$ corresponds to  $l\geq 1$.
Using Eq.\ \eqref{eqSTalphabeta}, Eq.\  \eqref{eqSLTarbD}  becomes
 \eq{
\lim_{\omega\rightarrow 0}L(\omega)T(\omega,\alpha, E, \beta, E')=\begin{cases}
-1 & \alpha=\beta=1,\\
0 &  \text{ Otherwise.}
\end{cases} \label{eqSLThighDcases}
}
Using Eqs.\ \eqref{eqSgvT}, \eqref{SeqLzeta}, and \eqref{eqSLThighDcases}, we obtain Eq.\ \eqref{eqSSDtheorem}. We are done with the proof of Theorem \ref{theoremHighDuniS} for all values of $a>0$. 
\end{proof}

\subsection{Levinson's theorem}




 Levinson's theorem can also be generalized to $D\geq 2$. We define the determinant of the infinite-dimensional matrix $\bm{S}(E)$ through the $n\rightarrow \infty$ limit  of the series  $\det[\bm{S}_n(E)]$, where $\bm{S}_n(E)$ is the $n\times n$ submatrix of $\bm{S}(E)$ in the subspace of  $\alpha=0,1,\dots n-1$:
 \eq{
 \det[\bm{S}(E)]\equiv \lim_{n\rightarrow \infty} \det[\bm{S}_n(E)]. \label{SeqdetSlimit}
 }
As $l$ increases,  $\phi_{l,q_l,E}(r,\bm{\theta})\sim (rk)^l$  vanishes increasingly fast close to the scattering center because of the centrifugal barrier. This implies that modes with high angular momentum ($\alpha \rightarrow \infty$) have trivial scattering amplitudes and the limit in Eq.\ \eqref{SeqdetSlimit} exists.

\begin{theorem} \label{theoremLevinsonEmitterD} 
 Define the winding phase $\Delta \delta$ of $\det[\boldsymbol{S}(E)]$  similarly to the 1D case in Theorem
  \ref{theoremLevinsonEmitter}.
Suppose $\ket{v_{\bm{k}}}=V(\bm{k})\ket{u}$ satisfies the properties listed in Theorem \ref{theoremHighDuniS}, and the dispersion relation is given by $\epsilon(k)=|\bm{k}|^a$, where $m>0$.  We have, in the absence of bright zero-energy eigenstates and bound states in the continuum, 
\eq{
\Delta \delta=\begin{cases}
\pi (N-N_B)+\pi \frac{a-D}{a} &  a> D,\\ 
\pi (N-N_B)&  a\leq D,
\end{cases}\label{eqThLevImp1}
}
where $N$ is the number of emitters and $N_B$ is the number of bound states.

\end{theorem}

Theorem \ref{theoremLevinsonEmitterD} can be proven through a procedure similar to the one used in  
 Theorem \ref{theoremLevinsonEmitter} for 1D systems.   Below, we  provide the extension of Theorem \ref{TheoremDetS} to arbitrary dimension $D$. The rest of the proof is quite straightforward and we omit it here. 
\begin{theorem}\label{TheoremDetShigher}

Define $J(\omega)=\det[\omega \mathbb{1}_N-\boldsymbol{K}^R-\boldsymbol{K}(\omega)]$, where $\bm{K}(\omega)$ is defined in Eq.\ \eqref{eqSKijD}. 
When $E$ is not equal to the energy of a bound state in the continuum, we have
 \eq{
\det[\boldsymbol{S}(E)]=\frac{J(E-i0)}{J(E+i0)}. \label{eqSSShigher} 
 }
 \end{theorem}

\begin{proof}
The proof follows the same procedure as the proof for Theorem \ref{TheoremDetS}. 
 Construct $\bm{A}_n$ as a $N\times n$ matrix and $\bm{A}_n^\dagger$ its Hermitian conjugate: 
\eq{
 \bm{A}_n&=\rho^{1/2}(E)\left(\begin{matrix}
 |v_{\alpha=0, E}\rangle, \dots
|v_{\alpha=n-1,E}\rangle
\end{matrix}\right). \quad 
} Then the $n\times n$ matrix 
$\bm{S}_n(E)$ for $E\neq E_c$ can be written as 
\eq{
\bm{S}_n(E)=\mathbb{1}_{n}-2\pi i \bm{A}_n^\dagger \boldsymbol{G}(E+i0^+) \bm{A}_n,
} 
where $\mathbb{1}_{n}$ is the identity matrix of dimension $n$ and $\bm{G}(\omega)\equiv(\mathbb{1}_N-\bm{K}(\omega)-\bm{K}^R)^{-1}$. Using the properties of determinant, the definition of $J(\omega)$, and
the identity 
 \eq{
 \boldsymbol{K}(E+i0^+)-\boldsymbol{K}(E+i0^-)&= \sum_{\alpha} \int_{0}^{+\infty}dE'\ \rho(E')|v_{\alpha,E'}\rangle \langle v_{\alpha, E'}|\left(\frac{1}{E+i0^+-E'}-\frac{1}{E-i0^--E'}\right),\nonumber\\
 &=-2\pi i \sum_{\alpha} \rho(E)|v_{\alpha,E}\rangle \langle v_{\alpha, E}|=-2\pi i \lim_{n\rightarrow \infty}\bm{A}_n\bm{A}_n^\dagger, \label{eqSKpKmhigher}
 }
  the r.h.s of Eq.\ \eqref{eqSSShigher} can be written as 
 \eq{
\frac{J(E-i0)}{J(E+i0)}=\det\left[\mathbb{1}_N-2\pi i \lim_{n\rightarrow 0}\bm{A}_n\bm{A}_n^\dagger\boldsymbol{G}(E+i0^+)\right].\label{eqdetratioexp1}
}
According to the matrix determinant lemma, given an invertible $N\times N$ matrix $-2\pi i\boldsymbol{G}(E+i0)$ and a  $N\times n(E)$ matrix $\bm{A}_n$, 
\eq{
\det[\mathbb{1}_N-2\pi i\bm{A}_n\bm{A}_n^\dagger\boldsymbol{G}(E+i0^+)]=\det[\mathbb{1}_{n}-2\pi i\bm{A}_n^\dagger\boldsymbol{G}(E+i0^+)\bm{A}_n] \label{eqSdoubledethigher}.
}
Using Eqs.\ \eqref{SeqdetSlimit}, \eqref{eqSnTnexp} and \eqref{eqdetratioexp1}, we see that the l.h.s and r.h.s of Eq.\ \eqref{eqSdoubledethigher} are equal, respectively, to $\det[\boldsymbol{S}(E)]$ and $\frac{J(E-i0)}{J(E+i0)}$  in the limit of $n\rightarrow \infty$. This is the end of the proof of Theorem \ref{TheoremDetShigher}.

\end{proof}

\section{Separable potential scattering  \label{SSecDelta}}
The purpose of this Letter is to demonstrate the principle that divergent density of states leads to nontrivial universal behavior of the $S$-matrix. To demonstrate that this principle is not limited to emitter-photon interactions, in this section, we show that the $S$-matrix has the same universal limit  in a particular class of potential scattering. To be specific, we study separable potential scattering.  Seperable potentials generalize delta-function potential scattering and serve as an important effective model to describe the long-distance behavior in many scattering systems. 

Assume any integer spatial dimension $D\geq1$.  The time-independent Schrodinger equation in momentum space is given by
\eq{
\epsilon(\bm{k})\psi(\bm{k})+ \,  \int d^D k'\ V(\bm{k'},\bm{k}) \psi(\bm{k'}) = E \psi(\bm{k}),
}
where the dispersion relation $\epsilon(\bm{k})$ is any of those being considered for emitter scattering in 1D in the main text or for arbitrary $D$ in Sec.\ \ref{SSecGen}.
 For local potentials, $V(\bm{k'},\bm{k})=V(\bm{k'}-\bm{k})$ depend only on the momentum difference  $\bm{k'}-\bm{k}$, where $\bm{k}~(\bm{k'})$ is the incoming (outgoing) momentum of the incident particle. 
For scattering with a separable potential,  the potential is non-local in real space and takes the form $\tilde{V}(\bm{z'},\bm{z}) =  g \tilde{v}(\bm{z'}) \tilde{v}(\bm{z})$, where $\bm{z}~(\bm{z'})$ is the incoming (outgoing) position of the incident particle and $\tilde{v}(\bm{z})$ is normalized such that $\left(\frac{1}{\sqrt{2\pi}}\right)^D\int d^D z \, \tilde{v}(\bm{z}) = 1$. Let $v(\bm{k})$ be the Fourier transform of $\tilde{v}(\bm{z})$; it is clear that $v(\bm{k})$ is uniformly continuous and $v(\bm{k}=\bm{0})=1$. The separable potential in momentum space takes the form $V(\bm{k'},\bm{k}) =   g \tilde{v}(\bm{k'}) \tilde{v}(\bm{k})$;  the time-independent Schrodinger equation  can then be written as 
\eq{
\epsilon(\bm{k})\psi(\bm{k})+ g\, v(\bm{k}) \int d^D k'  v(\bm{k'}) \psi(\bm{k'}) = E \psi(\bm{k}).
}
In potential scattering, the $T$-matrix can be solved from the Lippmann-Schwinger equation:
\eq{
T_{\text{sep}}(\omega,\bm{k},\bm{k'})=g\, v(\bm{k}) v(\bm{k'}) +g\, v(\bm{k}') \int d^D k'' \frac{v(\bm{k''}) }{\omega-\epsilon(\bm{k''})} T_{\text{sep}}(\omega,\bm{ k},\bm{k''}). \label{eqSTdelta}
}
Solving Eq.\ \eqref{eqSTdelta}, we obtain the solution for the  $T$-matrix:
 \eq{
T_{\text{sep}}(\omega,\bm{k},\bm{k'})&=\frac{ v(\bm{k}) v(\bm{k'})}{g^{-1}-K(\omega)}, \\
K_{\text{sep}}(\omega)&\equiv\int d^D k \frac{v(\bm{k})^2}{\omega-\epsilon(\bm{k})},
}
where the momentum dependence is simply captured by the momentum dependence of the potential. The scattering operator is related to the $T$-matrix through Eq.\ \eqref{SeqST}. 
Let us compare the $T$-matrix for separable potential scattering to single-particle emitter scattering with $K^R=0$ and $V_1(\bm{k})=v(\bm{k})$:
 \eq{
T(\omega,\bm{k},\bm{k'})&=\frac{ v(\bm{k}) v^*(\bm{k'})}{\omega-K(\omega)}, \\
K(\omega)&=\int d^D k \frac{|v(\bm{k})|^2}{\omega-\epsilon(\bm{k})}.
}
The similarity between $T_{\text{sep}}(\omega,\bm{k},\bm{k}')$ and $T(\omega,\bm{k},\bm{k}')$ allows the $S$-matrix universal limits to be similarly derived for separable potential scattering. To be specific, the zero-energy limit of the $S$-matrix behaves identically to the single-emitter scattering for any dispersion relation studied in earlier sections of the paper.  

The Levinson's theorem can also be generalized to separable potential scattering. 
Defining the $S$-matrix $\bm{S}_{\text{sep}}(E)$ similarly to $\bm{S}(E)$ in the case of emitter scattering, it is easy to derive a generalization of Theorems \ref{TheoremDetS} and \ref{TheoremDetShigher} to separable potential scattering.  In any dimension, we have
\eq{
\det[\bm{S}_{\text{sep}}(E)]=\frac{g^{-1}-K_{\text{sep}}(E-i0)}{g^{-1}-K_{\text{sep}}(E+i0)}.
}
The solutions of $g^{-1}-K_{\text{sep}}(\omega)=0$ correspond to bound state energies.  Define $\Delta \delta_{\text{sep}}$ as the difference of the scattering phase of $\det[\bm{S}_{\text{sep}}(E)]$ between the two ends of the continuum spectrum similarly to Eq.\ \eqref{eqDeltaPhi}. Following a proof similar to that of Theorem \ref{theoremLevinsonEmitter},  we have, for dispersion relation $\epsilon(k)=\pm |d|k^m$ in 1D,
\eq{
\Delta \delta_{\text{sep}}=-\pi N_B+\pi \frac{m-1}{m},
}
where  $N_B$ is the number of bound states.
For dispersion relation $\epsilon(\bm{k})=|\bm{k}|^a$ in integer dimension $D\geq 1$,
\eq{
\Delta \delta_{\text{sep}}=\begin{cases}
-\pi N_B+\pi \frac{a-D}{a} &  a> D,\\ 
-\pi N_B&  0<a\leq D.
\end{cases}
}
}



\end{document}